\documentclass[10pt, conference]{IEEEtran}
\IEEEoverridecommandlockouts
\usepackage{fancyhdr}
\usepackage{graphicx}
\usepackage[tight,footnotesize]{subfigure}
\usepackage{enumitem} 
\usepackage{float}
\usepackage{flushend}
\usepackage{multicol,multienum}
\usepackage{multirow}
\usepackage{diagbox}
\usepackage{hyperref}
\usepackage{bm}
\usepackage{cite}
\usepackage{graphicx}
\usepackage{amsthm}
\usepackage{tabularx}
\usepackage[ruled,vlined]{algorithm2e} 
\newtheorem{lemma}{Lemma}
\newtheorem{definition}{Definition}
\newtheorem{theorem}{Theorem}
\newtheorem{assumption}{Assumption}
\usepackage[tight,footnotesize]{subfigure}
\usepackage{url}
\usepackage{doi}
\usepackage{xcolor}
\usepackage{amssymb}
\usepackage{amsmath}
\usepackage{booktabs}
\usepackage{lscape}
\usepackage{lipsum}

\let\OLDthebibliography\thebibliography
\renewcommand\thebibliography[1]{
  \OLDthebibliography{#1}
  \setlength{\parskip}{3.4pt}
  \setlength{\itemsep}{4pt}
}

\newenvironment{sloppypar*}
{\sloppy\ignorespaces}
{\par}

\usepackage{tikz}

\usepackage{filecontents}

\begin{document}\sloppy

\date{}

\title{The Right to be Forgotten in Federated Learning: An Efficient Realization with Rapid Retraining}

\author{\IEEEauthorblockN{Yi Liu\textsuperscript{1}, Lei Xu\textsuperscript{2}, Xingliang Yuan\textsuperscript{3}, Cong Wang\textsuperscript{1}, and Bo Li\textsuperscript{4}}
	\IEEEauthorblockA{\textsuperscript{1}City University of Hong Kong; \textsuperscript{2}Nanjing University of Science and Technology; \\\textsuperscript{3}Monash University;
		\textsuperscript{4}The Hong Kong University of Science and Technology. \\ \{yiliu247, congwang\}@cityu.edu.hk, xuleicrypto@gmail.com, xingling.yuan@monash.edu, bli@cse.ust.hk}
\thanks{The authors have provided public access to their code and data at \href{https://github.com/yiliucs/federated-unlearning}{github.com/yiliucs/federated-unlearning}.}}
\maketitle
\thispagestyle{fancy} %
\lhead{} %
\chead{} %
\rhead{} %
\lfoot{} %
\cfoot{} %
\rfoot{\thepage} %
\renewcommand{\headrulewidth}{0pt} %
\renewcommand{\footrulewidth}{0pt} %
\pagestyle{fancy}
\rfoot{\thepage}
\begin{abstract}
In Machine Learning, the emergence of \textit{the right to be forgotten} gave birth to a paradigm named \textit{machine unlearning}, which enables data holders to proactively erase their data from a trained model. Existing machine unlearning techniques focus on centralized training, where access to all holders' training data is a must for the server to conduct the unlearning process. It remains largely underexplored about how to achieve unlearning when full access to all training data becomes unavailable. One noteworthy example is Federated Learning (FL), where each participating data holder trains locally, without sharing their training data to the central server. In this paper, we investigate the problem of machine unlearning in FL systems. We start with a formal definition of the unlearning problem in FL and propose a rapid retraining approach to fully erase data samples from a trained FL model. The resulting design allows data holders to jointly conduct the unlearning process efficiently while keeping their training data locally. Our formal convergence and complexity analysis demonstrate that our design can preserve model utility with high efficiency. Extensive evaluations on four real-world datasets illustrate the effectiveness and performance of our proposed realization.
\end{abstract}

\section{Introduction}\label{sec-1}

The paradigm of ``Machine Unlearning'' \cite{bourtoule2019machine} has attracted much attention recently. 
It has emerged from ``the right to be forgotten" \cite{voigt2017eu}, where individuals should be entitled to the right to have their private data removed from public directories. 
This regulated requirement has also set off a privacy-aware technology reform in the context of Machine Learning. 
From a high-level point of view, the techniques of machine unlearning enable data holders to proactively erase their data from a trained model that memorizes the data. 
Most of the existing machine unlearning techniques focus on the centralized setting~\cite{neel2021descent,koh2017understanding,baumhauer2020machine,ginart2019making,guo2019certified,izzo2020approximate,chen2021graph,wu2020deltagrad}, where all the training data samples are collected to a server, and later specific samples are requested to be erased via unlearning. 
%
%
However, this setting can be somewhat limited in practice because not all data holders are willing to share data for training due to increasing privacy concerns.


In this paper, we investigate the problem of machine unlearning in a more practical scenario, where data holders are collaboratively performing training and unlearning without sharing raw data. 
In particular, we target Federated Learning (FL) \cite{mcmahan2017communication,li2020federated,weng2021fedserving,zhang2020enabling,wang2019beyond,wang2020optimizing,wang2021device}, a widely adopted privacy-aware collaborative learning framework. In FL, data holders train a model from their local data samples, and the server only aggregates data holders' local model updates for data privacy considerations~\cite{chen2019round}. 
Despite the above benefit, not allowing the sharing of raw data also poses unique challenges for unlearning in FL. 


To erase data samples from a trained FL model, a straightforward approach is to simply train a new model from scratch on the remaining dataset (i.e., excluding the samples that need to be erased) following the original training procedure of FL.
However, this approach would be neither practical nor economical for real-world models that are usually trained for weeks from massive training samples.
Existing studies on centralized machine unlearning~
\cite{izzo2020approximate,baumhauer2020machine,ginart2019making,guo2019certified} are not directly applicable because they have to access the raw data during unlearning, which is prohibitive in FL in the first place.
Very recently, Liu \textit{et al.} \cite{liu2020federated} proposed an approach called \textit{FedEraser} that aims to  address this problem in FL. Unfortunately, this approach has two drawbacks. 
First, \textit{FedEraser} only initiates the unlearning algorithm among the clients that request data removal. Namely, other clients' models still hold the contributions of the data to be removed, which will later be aggregated in a global model. 
%
Second, their proposed algorithm utilizes the latest rounds of local model updates to approximate the current unlearned model update. However, the data samples to be removed can be used for training from the very beginning, and the utilized model updates may still consist of the trace of those samples.

Erasing data samples from a trained FL model is non-trivial. To better understand the problem, we first identify two unique challenges in the context of FL as follows. 
		
\noindent \textbf{(C1.)} \textbf{Training is interactive.} In FL, the client holds data locally and iteratively exchanges model updates with the server to learn a shared global model. In this context, it is challenging for the central server to erase the contribution of the client data to the model. The key reason is that the client continuously shares the knowledge learned from the local dataset with other clients via model aggregation. 
%
		
\noindent \textbf{(C2.)} \textbf{Training is incremental.} Training in FL is also an incremental procedure in which any given local update reflects all previous updates that occurred on the client. Suppose the global model is updated based on the local model updates that specific clients generate at a particular training round. In that case, all subsequent model updates will implicitly relate to the model updates of these clients. 

	    

\noindent\textbf{Our Contributions.} 
To address the above challenges, we deem that the best mechanism so far for data erasure in FL is to perform retraining among all the data holders, so as to completely eliminate the contributions of data samples to be removed.  
Then the key problem is how to design a rapid retraining approach in FL while preserving model utility. 
Specifically, we propose a distributed Newton-type model update algorithm, which follows the Quasi-Newton method \cite{xu2020second} and utilizes the first-order Taylor expansion to approximate the loss function trained by the local optimizer on the remaining dataset, i.e., the one excludes the removed data samples.
To further reduce the cost of retraining, we employ diagonal empirical Fisher Information Matrix (FIM) \cite{yao2021adahessian} to efficiently and accurately approximate the inverse Hessian vector to avoid the expensive cost of directly calculating it. 
To preserve model utility, we apply the momentum technique to the diagonal empirical FIM to alleviate the error caused by the approximation techniques and make the model convergence faster and more stable. The contributions of this paper are listed as follow:

\begin{itemize}
    \item We formally define the data erasure problem in federated learning and propose an efficient and effective retraining algorithm.
    Our proposed algorithm is model agnostic and can naturally be integrated to the framework of FL. 
    %
    %
    %
    
    \item We customize a retraining algorithm based on the diagonal empirical FIM for FL, by observing the first-order Taylor expansion of the loss function during the unlearning process. We introduce an adaptive momentum technique to reduce approximation errors in retraining.
    %
    \item We conduct a formal convergence analysis over our proposed algorithm and analyze the time and space complexity. Theoretical analysis results show that our algorithm converges to the true optimization path of the strongly convex objectives. 
    %
    
    \item We implement our design and conduct extensive experiments over four datasets in terms of the ability to achieve unlearning, efficiency, model performance, and parameter sensitivity. Compared to the baseline retraining, our approach can achieve a speedup of $9.1\times$ with an accuracy loss of 2.235 $\times 10^{-3}$ on the large dataset CelebA. 
\end{itemize}

\section{Related Work}


Our work is related to a line of studies on machine unlearning. Based on the methodology and target problems, they can be classified into two categories, exact unlearning methods and approximate unlearning methods.

\noindent \textbf{Exact Unlearning Methods.} These methods are designed to achieve machine unlearning in such a way that the produced models are effectively the same as the ones obtained with retraining. To be specific, they introduce a notion called certified data erasure, which is a theoretical guarantee that a model from which data is removed cannot be distinguished from a model that never observed the data to begin with.
%
Following this objective, some existing work \cite{brophy2021darerf,schelter2021hedgecut,ginart2019making,chen2019novel,baumhauer2020machine} proposes  unlearning methods for simple machine models, such as linear regression models, Random Forests models, and $k$-nearest neighbors models. However, these unlearning methods are not suitable for complex nonlinear models (e.g., deep learning models) \cite{izzo2020approximate}. 
On the other hand, these designs cannot directly be applied in FL, because they cannot erase the unlearned data samples' contributions (i.e., learned knowledge) shared by the clients in FL. 

\noindent  \textbf{Approximate Unlearning Methods.}
Intuitively, these methods aim to achieve higher efficiency of machine unlearning through the relaxation of the effectiveness and certifiability requirements~\cite{baumhauer2020machine,thudi2021necessity}. Most of them~\cite{wu2020deltagrad,neel2021descent,izzo2020approximate} ask the server to utilize historical gradients and model weights to efficiently approximate the gradients in the unlearning process. 
%
%
However, recent work \cite{zhu2020deep,huang2021evaluating} demonstrated that an attacker (e.g., a malicious central server) can exploit the client's local gradients to implement an attack to reconstruct the private training samples of the client. 
Therefore, the above methods may expose the client's private data to the server and are not desired for FL.
%
%
%
A recent design~\cite{liu2020federated} is closely related to our work, which aims to design an efficient retraining algorithm called \textit{FedEraser} to achieve unlearning in FL. 
As mentioned before, this work cannot fully erase the contributions of data samples to be removed from the trained model. 
Besides, it follows the above approximate methods, where the server utilizes cached historical gradients from clients to approximate the current unlearned gradient, which may lead to serious privacy leakage on data holders' data samples, e.g., gradient leakage attacks \cite{zhu2020deep,wang2019beyond}. 

 
\section{Preliminaries}\label{sec-2}
\subsection{Federated Learning Pipeline}\label{sec-2-1}
In the FL setting, we consider a server $\mathcal{S}$ and $K$ clients, participating in training a shared global model $\omega^*$ without sharing their raw private data. In this context, we assume that each client holds a local training dataset ${\mathcal{D}_k} = \{ {x_i},{y_i}\} _{i = 1}^{{n_k}}$, where $n_k$ denotes the number of samples. For the model parameters $\omega  \in {\mathbb{R}^d}$ and a local training dataset $\mathcal{D}_k$, let ${F_k}(\omega ,{x_i})$ be the loss function at the client, and let $f(\omega)$ be the loss function at the server. 
Accordingly, the pipeline of federated learning is defined as follows.

\noindent \textbf{\textit{Phase 1, initialization:}} First, the server selects a certain proportion $q$ of clients from all clients to participate in an FL learning task. Second, the server broadcasts the initialized global model $\omega_0$ to all clients, i.e., $\omega _0^k \leftarrow {\omega _0}$. 
	
\noindent \textbf{\textit{Phase 2, local training:}} For the $t$-th training round, each client trains the received global model $\omega _t^k$ on its own local dataset $\mathcal{D}_k$. On the client side, the goal is to minimize the following objective function: ${F_k}(\omega ) = \frac{1}{{{n_k}}}\sum\limits_{i \in {\mathcal{D}_k}} {{\ell_i}(\omega )},$ where $\mathcal{D}_k$ is the set of indexes of data samples on the client $k$. For a given federated learning task (such as image classification task), we typically take ${\ell_i}(\omega ) = \frac{1}{2}(x_i^T\omega  - {y_i})$ as the loss function of the local client, i.e., the loss of the prediction on training example $({x_i},{y_i})$ made with model parameters $\omega$. 
Then each client uploads its model updates $\Delta {\omega _k} = \eta \nabla {F_k}(\omega )$, where $\eta$ is the learning rate, to the server.


	
\noindent\textbf{\textit{Phase 3, aggregation:}} The server uses a model updates aggregation rule like FedAvg \cite{mcmahan2017communication}  to aggregate all the updates to obtain a new global model $\omega_{t+1}$. Specifically, on the server side, the goal is to minimize the following objective function:
\begin{equation}\label{eq-5}
	f(\omega ) = \sum\limits_{k = 1}^K {\frac{{{n_k}}}{n}{F_k}(\omega ),{F_k}(\omega ) = \frac{1}{{{n_k}}}} \sum\limits_{i \in {\mathcal{D}_k}} {{\ell_i}(\omega )}.
\end{equation}
Note that the above steps will be terminated until the global model reaches convergence.

\subsection{Machine Unlearning Pipeline}\label{sec-2-2}

In the context of ML, the emergence of the right to be forgotten gave birth to a paradigm called \textit{machine unlearning}, which enriches the life cycle of traditional ML pipelines by offering data deletion functions. To be specific, machine unlearning pipelines include how the model is trained, employed for inference, incremental update, and conducted a series of data deletions, while retraining from scratch \cite{mahadevan2021certifiable}.
%

\noindent \textbf{\textit{Training Stage:}} In the machine unlearning setting, we consider a service provider $\mathcal{S}$ and $K$ data owners, where $K$ data owners share the data with the service provider to form a initial training dataset $\mathcal{D}=\mathcal{D}_{init}$. 
Notably, the data owner can initiate a data deletion request to the service provider at any time during the life cycle of the pipeline to perform data deletion. Therefore, the initial training dataset $\mathcal{D}$ is not always constant, which is a subset of the initial training dataset, i.e., ${\mathcal{D}} \subseteq \mathcal{D}_{init}$. Specifically, if service provider deletes some data samples from $\mathcal{D}_{init}$, thus, we have: ${\mathcal{D}} \leftarrow \mathcal{D}\backslash {\mathcal{D}_m}$, where $\mathcal{D}_m$ is the deleted data subset. So we use $\mathcal{D}$ represents the currently available training dataset. The loss or objective function for a general machine learning model is defined as:
    \begin{equation}\label{eq-1}
        \mathop {\min }\limits_{\omega  \in {\mathbb{R}^d}} \mathcal{L}(\omega ), \text{where\quad}  \mathcal{L}(\omega ) = \frac{1}{n}\sum\limits_{i = 1}^n {{\ell _i}(\omega )}  + \frac{\lambda }{2}||\omega |{|_2},
    \end{equation}
where the above equation captures the average loss of the model, $\omega$ represents the model parameter, $\ell _i(\omega)$ represents the loss of the $i$-th data samples (e.g., cross-entropy loss), and $\lambda$ is the regularization constant parameter to prevent over-fitting. 
If we use mini-batch stochastic gradient descent (SGD) to update model parameter, then for $t = 1, \ldots ,T$, we have: ${\omega _{t + 1}} \leftarrow \omega  - \frac{\eta }{B}\sum\limits_{i = 1}^n \nabla  {\ell _i}({\omega _t}),$
where $\eta$ is the learning rate and $B$ is the mini-batch size, and model $\omega^*$ can be obtained by incremental update using the above equations. Service providers can deploy this model $\omega^*$ to provide inference services.
    
\noindent \textbf{\textit{Inference Stage:}} At this stage, the data holder can submit a service request to the service provider at any time, that is, upload the data sample $x_i$ to the server and expect to obtain the corresponding prediction result $y_i$. The service provider utilizes the currently available model $\omega$ to output the inference results to the data holders through the pipeline API. Furthermore, the data holder can also initiate data deletion requests to the service provider at any time, which will prompt the pipeline to proceed to the unlearning stage.
    
\noindent \textbf{\textit{Unlearning Stage:}} When the data holders finish requesting data deletion, the server performs the data deletion operation to obtain the remaining training dataset ${\mathcal{D}_u} \leftarrow \mathcal{D}\backslash {\mathcal{D}_m}$. The goal is to erase the contribution of the deleted data samples $\mathcal{D}_m$ from the current model $\omega$, that is, execute the \textit{unlearning algorithm} to ``unlearn'' the deleted dataset $\mathcal{D}_m$. For the \textit{unlearning algorithm}, a  naive method is \textit{retraining from scratch}, i.e., applying mini-batch SGD directly over the dataset $\mathcal{D}_u$ (we use $\omega^u$ to denote the corresponding model parameter). Then, run: $\omega _{t + 1}^u \leftarrow \omega _t^u - \frac{\eta }{{B - \Delta B_t}}\sum\limits_{i \in {\mathcal{D}_u}} {\nabla {\ell _i}(} \omega _t^u),$ where $\Delta B$ is the size of the subset removed from the $t$-th mini-batch. In this context, the goal is to minimize the following objective function: ${\mathcal{L}^u}(\omega ) = \frac{1}{{|{\mathcal{D}_u}|}}\sum\nolimits_{i \in {\mathcal{D}_u}} {{\ell _i}(\omega ).} $ 
    

\section{Problem definition}\label{sec-3}
\subsection{Formalizing the Problem of Federated Unlearning}\label{sec-3-2}
We now formalize the data erasure problem in FL. In our setting, the server $\mathcal{S}$ controls the whole training process by collecting and aggregating model updates uploaded by the clients and these clients ${\mathcal{C}} = \{ {C_1},{C_2}, \ldots ,{C_k}\} ,k \in \{ 1, \ldots ,K\}$ hold different local datasets ${\mathcal{D}_k}$. In particular, we enable the client to have the control to erase data, i.e., the client can request the server to delete some specific data samples. First, we introduce the concepts of federated data deletion operation. Here, we use the term ``learned client'' to refer to those clients that have not performed data erasure, and the term ``unlearned client'' to refer to those clients that perform data erasure. We assume that there are unlearned clients ${C_{{k_u}}} = \{ {C_1'},{C_2'}, \ldots ,{C_{{k_u}}'}\} ,{k_u} \in \{ 1, \ldots ,K\} \backslash {k_c}$ that initiate \textit{deletion} requests to the server, where $k_c$ is the index of learned client set $C_{k_c}$. Note that when $k_c=0$, it is a special case such that all clients initiate \textit{deletion} requests to the server.
\begin{definition}\label{defi-1}
\textit{(Federated Data Deletion Operation). A deletion $u$ is a pair $(\Omega, o )$ where $\Omega  \in \mathcal{D} $ is a data sample and $o$ is a deletion request. A deleted sequence $\mathcal{U}$ is a sequence $({u_1},{u_2},\ldots, u_i)$ where ${u_i} \in \mathcal{D} \times o$ for all $i$. For the $k$-th client, given a local dataset $\mathcal{D}_k$ and a deleted sequence $\mathcal{U}_k$, the deletion operation is defined as follows:}
\begin{equation}\label{eq-6}
\mathcal{D}_k^u = {\mathcal{D}_k} \circ {\mathcal{U}_k} \triangleq {\mathcal{D}_k}\backslash \{ \Omega \} _{i = 1}^{{R_k}},
\end{equation}
where $R_k$ is the number of deleted data samples. For example, an unlearned client deletes $R_k$ $({R_k} \ll {n_k})$ (i.e., $\mathcal{U}_k=\{u_1,u_2,\ldots,u_{R_k}\}$) data samples from the local dataset $\mathcal{D}_k$ according to the designed data delete operation to obtain the deleted dataset $\mathcal{D}_k^u$.
\end{definition}

By Definition \ref{defi-1}, we formalize the problem of federated unlearning as a game between two entities: the service provider $\mathcal{S}$ and the client set $\mathcal{C}$. The service provider can be a central aggregator that  collects model updates from various clients, where each client $C_k$ holds different local dataset ${\mathcal{D}_k}$. Service providers use the collected model updates to update and obtain the FL model $\omega^*$. Any unlearned client $C_{k_u}  \subseteq  \mathcal{C}$ can revoke the contribution of his data $\mathcal{D}_k^u \in \{\Omega\}_{i = 1}^{{R_k}}$ to $\omega^*$ by using the specified \textit{unlearning} algorithm $\mathcal{A}$.

To this end, the service provider has to erase the client's data contribution and \textit{retrain} any trained global models $\omega^*$ to produce $\omega^u$, that is, the contribution of the data $\mathcal{D}_k^u$ is not in the $\omega^u$ model. It is conceivable that the parameters of $\omega^*$ and $\omega^u$ are similar (despite stochasticity in learning), and it is desired for their performance (in terms of test accuracy) to be comparable. However, the fact that model $\omega^u$ was obtained by training on the remaining dataset from scratch provides a certificate to the data owner that their data share was indeed removed. Certified data erasure is defined as a theoretical guarantee that a model from which data is removed cannot be distinguished from a model that never observed the data to begin with. This conveys a very strong notion of privacy. Therefore, our goal is to make service providers revoke the contribution of specified data samples to the global model $\omega^*$ while maintaining this model's performance. The federated unlearning problem is defined as follows:

\begin{definition}\label{defi-5}
\textit{(Federated Unlearning). Let $\omega^*$ denote the global model produced by the federated learning pipeline, and $\omega^u$ denote the global model produced by the federated unlearning pipeline. For a given threshold $\varepsilon  \geqslant 0$, the federated unlearning problem is defined as: there exists an unlearning algorithm $\mathcal{A}$ that makes the following equation hold:}
 \begin{equation}\label{eq-7}
||T({\omega ^*}) - {T_\mathcal{A}}({\omega ^u})|| \leqslant \varepsilon, 
\end{equation}   
where $T( \cdot )$ denotes the distribution of models learned using unlearning algorithm $\mathcal{A}$ and $\varepsilon$ is a small positive number. $\mathcal{A}$ achieves unlearning when these two distributions are $\varepsilon$--identical. In other words, in our work, this evidence (i.e., the output of unlearning algorithm $\mathcal{A}$) takes the form of a training algorithm, which if implemented correctly, guarantees that the parameter distributions of $\omega^*$ and $\omega^u$ are $\varepsilon$--identical.
\end{definition}  

\subsection{Goals of Federated Unlearning}\label{sec-3-3}
As discussed in Section \ref{sec-2-2}, the simple and effective way to achieve federated unlearning is to retrain the global model from scratch on the remaining training dataset. However, such a method is unrealistic for an FL system due to large computing resources and communication overhead. Moreover, to comply with data regulations such as GDPR \cite{voigt2017eu} for continuous data deletion requests, service providers need to retrain the global model frequently. 
In this paper, we consider this method as a baseline method to compare with our solution. 
To this end, our proposed solution needs to meet the following goals:

\noindent \textbf{(G1.)} \textbf{Zero Contribution Guarantees:} Similar to the baseline, any new unlearning solution should ensure that the erased data has zero contribution to the unlearned model, i.e., the erased data does not influence the model parameters. 
		
\noindent \textbf{(G2.)} \textbf{Privacy Guarantee:} Any new unlearning solution should maintain the privacy of the client in the FL settings. Any attacker, including the server, cannot easily utilize gradient information to implement gradient leakage attacks \cite{zhu2020deep}.
		
\noindent \textbf{(G3.)} \textbf{Comparable Accuracy:} It is conceivable that there is a trade-off between the performance of the unlearning solution and the accuracy of the model, which indicates that the accuracy of the model would decrease. Even if there is no component of the approach that explicitly promotes high accuracy, any unlearning solution should strive to introduce a small accuracy gap in comparison to the baseline for any number of data samples unlearned.
		
\noindent \textbf{(G4.)} \textbf{Reduced Retraining Time:} Additionally, no matter how many data samples need to be unlearned, any new unlearning solution should be faster than the baseline (or other state-of-the-art unlearning approaches). Namely, any new unlearning solution should not introduce additional computation and communication overhead to the original training procedure, which is already expensive  \cite{bourtoule2019machine}. 

		
		
\noindent \textbf{(G5.)} \textbf{Model Agnostic:} The  unlearning solution should be general, which means that it can be applied to any model. It should also provide the aforementioned guarantees for the system settings of varying nature and complexity \cite{bourtoule2019machine,ginart2019making}.

\begin{figure}[!t]
 \centering
 \includegraphics[width=1\linewidth]{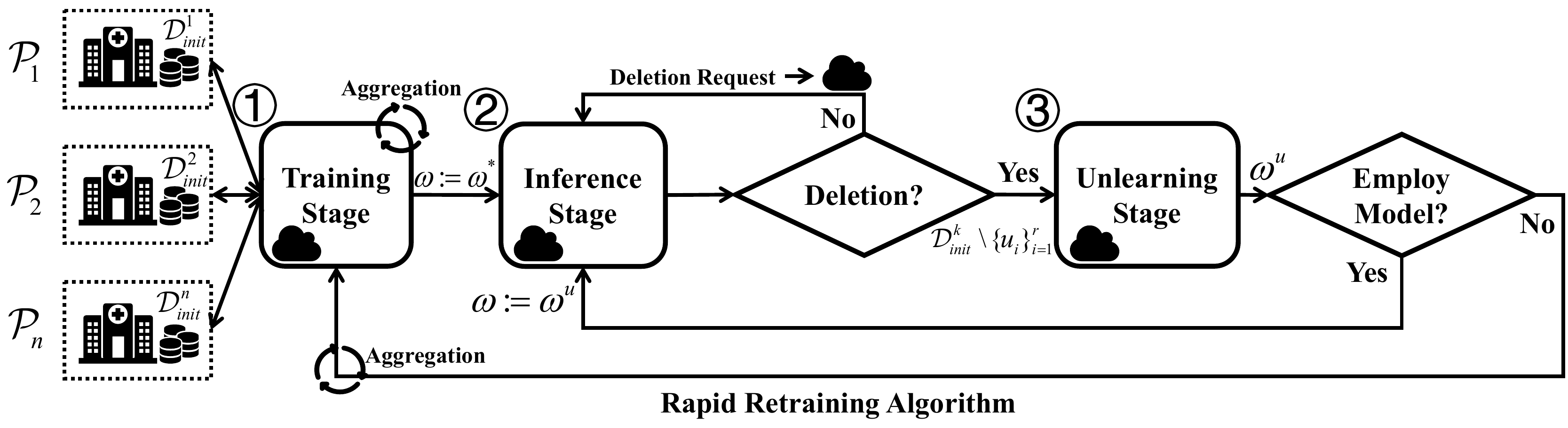}
 \caption{Federated unlearning pipeline with the three stages of training, inference, and unlearning.}
  \label{fig:overview}
  \vspace{-0.5cm}
\end{figure}
\section{Our Design}\label{sec-4}
\subsection{Federated Unlearning Pipeline}\label{sec-4-1}
In this section, we present a federated unlearning pipeline, as illustrated in Fig. \ref{fig:overview}. Specifically, the pipeline includes three stages: \textit{training stage, inference stage, and unlearning stage.} The training stage and the inference stage is the same as those in FL systems (refer to Section \ref{sec-2-1} for details). 
The main difference from FL is that this pipeline adds the two functions, i.e., data deletion and rapid retraining, both of which are included in the unlearning stage. Thus, the details of the unlearning stage are as follows: 

\noindent \textbf{\textit{Unlearning Stage:}} \textbf{\textit{(a) Data Deletion}}: Assuming that the $k_u$-th unlearned client initiates several data deletion requests $\mathcal{U}_k=\{u_1,u_2,\ldots,u_{R_k}\}$ to the server before the start of the $t$-th round of training, these clients need to perform data deletion operations, where the remaining local dataset is denoted as $\mathcal{D}_k^u$. \textbf{\textit{(b) Rapid Retraining}}: Then all clients execute the proposed unlearning algorithm (see below for more details) on the remaining local dataset $\mathcal{D}_k^u$ to achieve unlearning. Specifically, to eliminate the influence of unlearned data samples on the FL model, we perform a \textit{retrain} operation to make the model forget the knowledge represented by these unlearned data samples. In this context, the naive solution (i.e., retraining from scratch) is to apply mini-batch stochastic gradient descent (SGD) directly over the remaining local training dataset $\mathcal{D}_k^u$ (we use $\omega_t^{k_u}$ to denote the corresponding model parameter), i.e., for $k$-th client, the loss function can be defined as follows:
\begin{equation}\label{eq-8}
{F_{k_u}}(\omega ) = \frac{1}{{B - \Delta B_t}}\sum\limits_{({x_i},{y_i}) \in \mathcal{D}_k^u} {{\ell _i}(\omega )},
\end{equation}
where $B$ is the mini-batch size and ${\Delta {B_t}}$ is the size of the subset removed from the $t$-th mini-batch. Thus, each client updates its local model by using: $\omega _{t + 1}^{{k_u}} = \omega _t^{{k_u}} - \frac{\eta }{{B - \Delta {B_t}}}\sum\limits_{({x_i},{y_i}) \in \mathcal{D}_k^u} {\nabla {F_k}(\omega _t^u)} ,$ where $\eta $ is the learning rate. After executing the local unlearning step, all clients upload model updates to the server so that the server can aggregate these updates to obtain a new unlearned global model. Thus, we have the following execution step. \textbf{\textit{(c) Aggregation:}} Recall that, to achieve unlearning, all unlearned clients and all learned clients perform local retraining and upload the updates of the unlearned model to the server. Then the server uses the aggregation rules to aggregate the model updates of the unlearned client and other clients. Such an operation can meet the goals \textbf{G1} and \textbf{G2}. Suppose we use the classic aggregation rule (i.e., FedAvg) to aggregate these updates, the formal definition of this aggregation rule is as follows:
\begin{equation}\label{eq-9}
    \omega _{t+1}^u = \frac{1}{{{K_c}\sum {{p_{{k_c}}}} }}\sum\limits_{{k_c}} {{p_{{k_c}}}\omega _t^{{k_c}}}  + \frac{1}{{{K_u}\sum {{p_{{k_u}}}} }}\sum\limits_{{k_u}} {{p_{{k_u}}}\omega _t^{{k_u}}}.
\end{equation}
Note that user-defined aggregation rules can also be used in our pipeline, which do not affect the operations of it.

\subsection{Federated Rapid Retraining}\label{sec-4-2}
Obviously, the baseline is a time-consuming and resource-consuming unlearning solution, which is undesirable for real-world large FL systems. This motivates us to seek a time-saving and energy-efficient rapid retraining solution. To this end, we follow the Quasi-Newton methods \cite{xu2020second} and use the first-order Taylor approximation technique to propose an efficient method. Specifically, for the $k$-th client, let $\nabla {F_{{k_u}}}({\omega}) = 0$ around $\omega^*$, we obtain:
\begin{equation}\label{eq-10}
\nabla {F_{{k_u}}}({\omega ^*}) + {H_{{k_u}}}({\omega ^*})({\omega ^{u}} - {\omega ^*}) \approx 0,
\end{equation}
where ${\omega ^{u}} = \arg \min \nabla {F_{{k_u}}}(\omega )$ is a unique global minimum (i.e., global optimal unlearned model), ${H_{{k_u}}}({\omega ^*}) = {\nabla ^2}{F_{{k_u}}}({\omega ^*})$ is the Hessian matrix, and $\omega^*$ is a minimizer for $F_k(\omega)$. According to the first-order optimality condition, we have: $\omega _{t + 1}^{{k_u}} = \omega _t^{{k_u}} + \frac{1 }{{B - \Delta {B_t}}}H_{{k_u}}^{ - 1}{\Delta _k},$ where $\Delta$ is the local gradient, i.e., $\Delta_{k_u}  = \nabla {F_k}(\omega _t^{{k_u}},\mathcal{D}_k^u)$, and $H_{k_u} = {\nabla ^2}{F_k}(\omega _t^{{k_u}},\mathcal{D}_k^u)$. In this way, we can utilize \textit{Newton-type} update strategy to efficiently achieve unlearning goals. However, the calculation of the inverse Hessian-vector is still computationally expensive. Therefore, we further explore how to efficiently calculate the inverse Hessian matrix $H_{k_u}^{-1}$.

\begin{algorithm}[!t]
\footnotesize
	\caption{Federated Rapid Retraining Algorithm}\label{al-1}
	\LinesNumbered 
	\KwIn{Local training dataset $\mathcal{D}_k=\{ {x_i},{y_i}\} _{i = 1}^{n_k}$, model $\omega$, mini-batch size $B$,  unlearned dataset ${\cal U}_k$, learning rate ${\eta}$, and local epoch $E_{local}$.}
	\KwOut{Unlearned global model $\omega^u$.}
	\textbf{Deletion Operation:}\hfill $\rhd$Run on client $k$\\
	\ForEach{unlearned client \textbf{in parallel}}{
	Perform batch data deletion operations, i.e., ${\cal D}_k^u \leftarrow {{\cal D}_k}\backslash {{\cal U}_k}$\;
	}
	\textbf{Rapid Retraining Stage:}\hfill $\rhd$Start unlearning stage\\
	\textbf{Server Executes:}\\
	Reinitialize the global model and send it to all clients, i.e., $\omega _0^{k_u} \leftarrow \omega$\;
	\textbf{Client Executes:}\\
	\For{all clients \textbf{in parallel}}{
	\ForEach{unlearned client \textbf{in parallel}}{
		\For{$t = 0,1, \ldots ,T$}{
       LocalTraining $({{k_u}},\omega _t^{{k_u}},{E_{local}},{\cal D}_k^u)$\;
	}
}
	\ForEach{normal client \textbf{in parallel}}{
	\For{$t = 0,1, \ldots ,T$}{
       LocalTraining $({{k_c}},\omega _t^{{k_u}},{E_{local}},{\cal D}_k)$\;}
	}
	}
\textbf{Server Executes:}\\
Update unlearned global model parameter ${\omega _{t + 1}^u}$\; 

\textbf{LocalTraining} $({{k}},\omega ,{E_{local}},{\cal D})$ $\rhd$Run on client $k$\\
\ForEach{local epoch $i$ from 1 to $E_{local}$}{
\For{batch $b \in {B}$}{
${\cal D}_k^u \leftarrow {{\cal D}_k}\backslash {\cal U}_k$ \hfill\% Current unlearned dataset\;
${\Delta _k} \leftarrow \nabla {\ell _k}(\omega _t^{{k_u}},{\cal D}_k^u)$ \hfill \% Current step gradient\;
${\Gamma_k} \leftarrow diag(\Gamma_k)$ \hfill \% Current step estimated diagonal FIM\;
Update ${{\bar G}_t}$ based on Eq. \eqref{eq-15}\;
Update $m_t$, $v_t$ based on Eq. \eqref{eq-15}\;
$\omega _{t + 1}^{{k_u}} = \omega _t^{{k_u}} - \frac{\eta }{{B - \Delta {B_t}}}{m_t}/{v_t}$\;
}
}
\Return $\omega^u$.  
\end{algorithm}	
To address this problem, recent work \cite{wu2020deltagrad,liu2020federated} uses the limited-memory Broyden Fletcher Goldfarb Shanno (L-BFGS) algorithm \cite{matthies1979solution,berahas2016multi,bollapragada2018progressive} by leveraging the historical parameter-gradient pair $\{ (\omega _t^k,\nabla F_k(\omega _t^k))\} _{t = 1}^T$ stored in the unlearned client to approximate $H_{k_u}^{-1}$ for each of $T$ iterations. In particular, L-BFGS fits the Hessian matrix through the first $m$ historical parameter-gradient pairs without explicitly constructing and storing the approximation of the Hessian matrix or its inverse matrix, and its time and space complexity is $\mathcal{O}(mn)$. Nevertheless, L-BFGS algorithms can just efficiently solve the Hessian approximation problem only when the model is small (i.e., generally the model parameter is less than $10^4$, but they cannot be directly applied to the setting of large models (e.g., ResNet \cite{he2016deep}). Furthermore, if the server stores historical gradients and parameters, it will incur privacy disclosure risks \cite{zhu2020deep} to the clients.

Motivated by the limitations mentioned above, we aim to answer the following question: \textit{how to efficiently approximate the inverse Hessian matrix without utilizing the historical parameter-gradient pair $\{ (\omega _t^k,\nabla F_k(\omega _t^k))\} _{t = 1}^T$ in FL?} In this paper, we propose a low-cost Hessian approximation method, i.e., diagonal empirical Fisher Information Matrix (FIM)-based approximation method to efficiently approximate Hessian matrix. So the unlearning update rule can be rewritten as follows:
\vspace{-0.4cm}
\begin{equation}\label{eq-11}
\omega _{t + 1}^{{k_u}} = \omega _t^{{k_u}} - \frac{1 }{{B - \Delta {B_t}}}\Gamma _{{k_u}}^{ - 1}{\Delta _{{k_u}}},
\end{equation}
where $\Gamma _k$ is the FIM. Then, we show that the outer product of the gradient is an asymptotically unbiased estimate of the Hessian matrix and we define the outer product matrix of the gradient at $\omega_t^{k_u}$ as follows:
\begin{equation}\label{eq-12}
H_t^{{k_u}} \approx \Gamma _t^{{k_u}} = \frac{1}{{B - \Delta B}}\sum\limits_{({x_i},{y_i}) \in \mathcal{D}_k^u} {\nabla {F_k}({x_i},\omega _t^{{k_u}})} \nabla {F_k}{({x_i},\omega _t^{{k_u}})^ \top },
\end{equation}
where $H_{t}^{k_u}$ is the empirical expectation of the outer product as an approximation to the Hessian at $\omega_t^{k_u}$. Recall that, since the cross-entropy loss we use is negative log-likelihood, it is not difficult to obtain ${\mathbb{E}_{{x_i} \in \mathcal{D}_k^u}}[\nabla {F_{{k_u}}}({x_i},{\omega ^*})\nabla {F_{{k_u}}}{({x_i},{\omega ^*})^{\top}}]$, where $\omega ^*$ is the true parameter (i.e., the Softmax distribution of the local model) obtained in the form of FIM. According to the definition of two equivalent methods for calculating the FIM \cite{friedman2001elements,ly2017tutorial}, the above mentioned equation can also be written as ${\mathbb{E}_{{x_i} \in \mathcal{D}_k^u}}[{\nabla ^2}{F_{{k_u}}}(x,{\omega ^*})]$. Thus, the lemma about the approximation error $\epsilon_t$ is as follows:
\begin{lemma}\label{le-1}
(Upper bound on $\epsilon_t$). Let $\epsilon_t$ be the approximate error of the FIM approximation of the Hessian matrix, then when $t \to \infty $, the following equation holds:
\begin{equation}\label{eq-13}
\epsilon_t = {\mathbb{E}_{(y|x,{\omega ^*})}}[H_t + \Gamma _t] \to 0.
\end{equation} 
\end{lemma}
\begin{proof}\label{proof-1}
The proof of Eq. \eqref{eq-13} can be equivalent to proving that the following equation holds: ${\mathbb{E}_{(y|x,{\omega ^*})}}{H_t} =  - {\Gamma _t}$. First, suppose we have a model parameterized by a parameter vector $\omega$, which models the distribution $p (x|\omega)$. To learn model $\omega$, we need to maximize the likelihood $p (x|\omega)$ wrt. parameter $\omega$. 
Now, 
since the Hessian matrix of the log-likelihood is given by the Jacobian matrix of its gradient, then we have ${H_{{{\log }_{p(x|\omega )}}}} =  \frac{{{H_{p(x|\omega )}}}}{{p(x|\omega )}} - (\frac{{\nabla p(x|\omega )}}{{p(x|\omega )}}){(\frac{{\nabla p(x|\omega )}}{{p(x|\omega )}})^{\top}}.$

By the above results, it is easily to get
\begin{align*}
\epsilon_t &\textbf{}= {\mathbb{E}_{(y|x,{\omega ^*})}}[H_t + \Gamma _t]\\
&= \mathop \mathbb{E}\limits_{p(x|\omega )} [\frac{{{H_{(x|\omega )}}}}{{p(x|\omega )}} - (\frac{{\nabla p(x|\omega )}}{{p(x|\omega )}}){(\frac{{\nabla p(x|\omega )}}{{p(x|\omega )}})^{\top}}]+{\mathbb{E}_{(y|x,{\omega ^*})}}[\Gamma _t]\\
&=\int {{{{H_{p(x|\omega )}}}}dx - \mathop \mathbb{E}\limits_{p(x|\omega )} [\nabla \log p(x|\omega )\nabla \log p{{(x|\omega )}^{\top}}-\Gamma _t]}\\
&=H \int p(x|\omega)dx=H_1=0.
\end{align*}  
\vspace{-0.4cm}
\end{proof}
Above result implies that we can use ${\mathbb{E}_{{x_i} \in \mathcal{D}_k^u}}[\nabla {F_{{k_u}}}({x_i},{\omega ^*})\nabla {F_{{k_u}}}{({x_i},{\omega ^*})^{\top}}]$, where $\omega ^*$ as an asymptotically unbiased estimate of ${\mathbb{E}_{{x_i} \in \mathcal{D}_k^u}}[{\nabla ^2}{F_{{k_u}}}(x,{\omega ^*})]$ when $\omega$ gradually converges to $\omega^{u}$ during training.

Additionly, to significantly reduce the computational complexity and storage while maintaining the accuracy of the approximator $\Gamma $, in DNN we leverage the diagonalization technique to approximate FIM \cite{becker1988improving,zheng2017asynchronous}, i.e., $diag(\Gamma) \approx {\Gamma}$, where $diag(\Gamma)$ is the diagonal matrix of $\Gamma$. Specifically, we only need to store the diagonal elements of $\Gamma$ and make all other elements zero. 
Although this method can approximate the Hessian matrix efficiently, it may be affected by the variance generated by the iterative update of FL, which may lead to higher approximation errors and unstable convergence. To address this issue, this paper follows the momentum techniques used in Adam \cite{kingma2014adam} and RMSProp \cite{tieleman2012lecture} to easily apply momentum to the Hessian diagonal since it can achieve more stable and faster convergence. More specifically, let ${{\bar G}_t}$ denote the Hessian diagonal with momentum, and the first and second-order moments ($m_t$ and $v_t$) for Hessian momentum are computed as follows:
\begin{equation}\label{eq-15}
{m_t} = \frac{{(1 - {\beta _1})\sum\nolimits_{i = 1}^t {\beta _1^{t - i}{\Delta _i}} }}{{1 - \beta _1^t}},{v_t} = \sqrt {\frac{{(1 - {\beta _2})\sum\nolimits_{i = 1}^t {\beta _2^{t - i}{\Gamma _i}{\Gamma _i}} }}{{1 - \beta _2^t}}} ,
\end{equation}
where ${v_t} = {{\bar G}_t}$, $0 < \beta _1^t < \beta _2^t < 1$ are the first and second moment hyperparameters that are also used in Adam. Therefore, refer to Eq. \eqref{eq-10}--\eqref{eq-15}, Eq. \eqref{eq-11} can be rewritten as follows: $\omega _{t + 1}^{{k_u}} = \omega _t^{{k_u}} - \frac{1 }{{B - \Delta {B_t}}}{m_t}/{v_t}.$
Furthermore, other normal clients perform local training, i.e., run the following update rule: $\omega _{t + 1}^{{k_u}} = \omega _t^{{k_u}} - \frac{1 }{B}{m_t}/{v_t}.$

We present our federated rapid retraining algorithm, as shown in Algorithm \ref{al-1}. As seen in lines 1 to 3, all unlearned clients perform a mini-batch data deletion operation, i.e., split the deleted data into small batches of size $B <= R_k$ and sequentially for each of them. As a result, each unlearned client obtains a locally deleted dataset $\mathcal{D}_k^u$. Such a method will result in multiple smaller second-order Fisher information correction steps, thereby obtaining a more effective FL model at the cost of efficiency. Then all clients start the unlearning stage and perform local training on the local dataset. In particular, the unlearned client performs local training on the locally deleted dataset $\mathcal{D}_k^u$, while the normal clients perform local training on the original dataset $\mathcal{D}_k$. This is described in lines 9 to 14 and lines 18 and 25. As seen in line 16, the server uses aggregation rules such as FedAvg to iteratively aggregate the model updates uploaded by the client.
  
\section{Theoretical Analysis}\label{sec-5}
\subsection{Convergence Analysis} \label{sec-5-1}
We first introduce assumptions and then provide the convergence rate of Algorithm \ref{al-1} for strongly convex objectives in Theorem \ref{the-1}. The norm used throughout the rest of the paper is $\ell_2$ norm.


\begin{assumption} \label{assum-2}
(Bounded gradients). For any model parameter $\omega$ and in the sequence $[{\omega _0},{\omega _1}, \ldots ,{\omega _t}, \ldots ]$, the norm of the gradient at every sample is bounded by a constant $\varepsilon_0^2$, i.e., $\forall \omega ,i,$, we have: $||\nabla {f_i}(\omega )|| \leqslant {\varepsilon_0^2}$.
\end{assumption}

\begin{assumption}\label{assum-3}
(Lipschitz continuity). We assume that the function $F:{\mathbb{R}^d} \to \mathbb{R}$ is $L$-Lipschitz continuous, i.e., $\forall \omega_1, \omega_2$, the following equation holds:
\begin{equation}
|F({\omega _1}) - F({\omega _2})| \leqslant L||{\omega _1} - {\omega _2}||.
\end{equation}
\end{assumption}
\begin{assumption}\label{assum-4}
(Strong convexity and smoothness). $F(\omega)$ is $\mu$-strongly convex and $\rho $-smooth with positive coefficient $\mu$ if $\forall {\omega _1},{\omega _2}$, the following equations hold:
\begin{equation}
F({\omega _1}) \geqslant F({\omega _2}) + \nabla F{({\omega _2})^ \top }({\omega _1} - {\omega _2}) + \frac{\mu }{2}||{\omega _1} - {\omega _2}|{|^2}.
\vspace{-0.5cm}
\end{equation}
\begin{equation}
F({\omega _1}) \leqslant F({\omega _2}) + \nabla F{({\omega _2})^ \top }({\omega _1} - {\omega _2}) + \frac{\rho }{2}||{\omega _1} - {\omega _2}|{|^2}.
\end{equation}
\end{assumption}
\begin{assumption}\label{assu-5}
\textit{The function $F:{\mathbb{R}^d} \to \mathbb{R}$ is twice continuously differentiable, $\rho$-smooth, and $\mu$-strongly convex,  $\mu >0$, i.e.,}
\begin{equation}\label{eq-14}
\mu I \leqslant {\nabla ^2}f(\omega ) \leqslant \rho I,
\end{equation}
\textit{where $I \in \mathbb{R}^d$ and ${\nabla ^2}f(\omega )$ is the Hessian of gradient.}
\end{assumption}
\begin{theorem}\label{the-1}
Suppose the objective function is strongly convex and smooth, and Assumption \ref{assum-2} holds. Thus, we have:
\begin{equation}\label{eq-16}
f(\omega _{t + 1}^u) - f({\omega ^*}) \leqslant  \varepsilon,
\end{equation}
where $\varepsilon  =  - \frac{\mu }{{2{\rho ^2}}}\varepsilon _0^2$.
\end{theorem}
\begin{proof}\label{proof-2}
First, since $f(\omega)$ is $\mu$-strongly convex, Eq. \eqref{eq-16} is equivalent to the following equation: 
\begin{equation*}f(\omega _t^u - {\eta _t}\Delta {\omega _t}) \leqslant f(\omega _t^u) - {\eta _t}\Delta _t^{\top}\Delta {\omega _t} + \frac{{{\eta _t}^2\rho ||\Delta {\omega _t}|{|^2}}}{2},
\end{equation*}
where $\Delta {\omega _t} = H_t^{ - 1}{\Delta _t}$ and $\Delta_t=\nabla f(w_t^u)$. Let $\phi ({\omega _t}) = \sqrt {\Delta _t^{\top}H_t^{ - 1}{\Delta _t}}$, thus, we have: $\phi {({\omega _t})^2} = \Delta \omega _t^{\top}(H_t)^{\top}\Delta {\omega _t} \geqslant \mu ||\Delta {\omega _t}|{|^2}$. Then we can obtain: 
\begin{equation*}
f(\omega _t^u - {\eta _t}\Delta {\omega _t}) \leqslant f(\omega _t^u) - {\eta _t}\phi {({\omega _t})^2} + \frac{\rho}{{2\mu}}\eta _t^2\phi {({\omega _t})^2}.
\end{equation*}
We assume that the step size ${\eta _t} = \frac{\mu }{\rho }$, thus, $f$ decreases as follows: \begin{equation*}
f(\omega _t^u - {\eta _t}\Delta {\omega _t}) \leqslant f({\omega ^*}) - \frac{1}{2}{\eta _t}\phi {({\omega _t})^2}.
\end{equation*}
If assumption \ref{assu-5} holds, i.e., $\mu  \leqslant {H_t} \leqslant \rho $ holds, thus, we have: $\phi {({\omega _t})^2} = {\eta _t}\Delta _t^{\top}H_t^{ - 1}{\Delta _t} \geqslant \frac{1}{\rho }||{\Delta _t}|{|^2}$. It can be seen from Lemma 1 that if we use FIM to approximate Hessian, the above inequality still holds. When $t \to \infty $, we have ${\omega ^*} \leftarrow {\omega _t}$, thus, the following inequality holds:
\begin{equation*}
f(\omega _t^u - {\eta _t}\Delta {\omega _t}) - f({\omega ^*}) - \frac{1}{{2\rho }}{\eta _t}||{\Delta _t}|{|^2} =  - \frac{\mu }{{2{\rho ^2}}}||{\Delta _t}|{|^2} \leqslant  - \frac{\mu }{{2{\rho ^2}}}\varepsilon _0^2.   
\end{equation*}
Furthermore, as long as Lemma 1 is true, the above inequality can be established regardless of whether the diagonal FIM or FIM is used to approximate the Hessian.
\end{proof}


\subsection{Complexity Analysis}\label{sec-5-2}
Here we give a formal analysis of the time and space complexity of the proposed algorithm. We define the number of parameters of the model as $p$, the number of samples is $n$, and the dimension of the feature is $d$.

\noindent \textbf{Time Complexity:} 
Let $f(p)$ be the time complexity of forward propagation, then the time complexity of one step backpropagation is at most $5f(p)$ \cite{griewank2008evaluating}, so the total complexity of computing the derivative of each training sample is $6f(p)$. Thus, the total time complexity of the baseline at the step is $6f(p)[{k_u}(B - \Delta B) + {k_c}B]$. If we use a block diagonal structure with a block size of $b$ and iterate over $(B - \Delta B)$ samples in the diagonal estimation, the total computational complexity of the proposed algorithm at the step is $6f(p)[{k_u}(B - \Delta B) + {k_c}B] + k\mathcal{O}(b(B - \Delta B)d)$. Suppose there are $T_b$ ($T_u$) iterations in the retraining process. Then the running time ${\rm T}_b$ of baseline method will be $6f(p)[{k_u}(B - \Delta B) + {k_c}B]{T_b}*{t_b}$. The proposed algorithm's total running time ${\rm T}_u$ is $\big[6f(p)[{k_u}(B - \Delta B) + {k_c}B] + k\mathcal{O}(b(B - \Delta B)d)\big]{T_u}*{t_u}$. Accordingly, let $v$ denote the speed-up factor which is defined as follows:
\begin{equation}
v = {(\frac{{{{\rm T}_u}}}{{{{\rm T}_b}}})^{ - 1}} = {[\frac{{{T_u}*{t_u}}}{{{T_b}*{t_b}}}(1 + \frac{{k\mathcal{O}(b(B - \Delta B)d)}}{{6f(p)[{k_u}(B - \Delta B) + {k_c}B]}})]^{ - 1}},
\end{equation}
where $t_u$ and $t_b$ are the running time of one round of training of Algorithm \ref{al-1} and the baseline, respectively. The cost of this estimation is one diagonal FIM (to compute $diag(\omega)$), which is equivalent to one gradient backpropagation~\cite{yao2020pyhessian,yao2018hessian}, i.e., $k\mathcal{O}(b(B-\Delta B)d) \approx 5kf(p)$, thus, we have: 
\begin{equation}
v \approx {[\frac{{{T_u}*{t_u}}}{{{T_b}*{t_b}}}(1 + \frac{{5k}}{{6[{k_u}(B - \Delta B) + {k_c}B]}})]^{ - 1}}.    
\end{equation}

\noindent \textbf{Space Complexity:}
In the proposed algorithm, we use the diagonal technique to approximate the Fisher information matrix and update it with the mini-batch gradients. Specifically, as we employ a block-diagonal structure with blocks of size $b$, it needs a memory of size $\mathcal{O}(db)$ to compute the estimated diagonal Fisher information \cite{yao2021adahessian}. As shown in our experiments, this allows us to support fairly large models and sample sets.

\begin{table}[!t]
\centering
\scriptsize
\caption{Datasets.}
\label{tab-1}
\begin{tabular}{@{}ccccc@{}}
\toprule
Dataset & Dimensions & Classes & Training & Test \\ \midrule
MNIST & 784 & 10 & 60,000 & 10,000 \\
Fashion-MNIST & 784 & 10 & 60,000 & 10,000 \\
CIFAR-10 & 1024 & 10 & 50,000 & 10,000 \\ 
CelebA & 784 & 2 & 160,000 & 40,000\\ \midrule
\end{tabular}
		\vspace{-0.5cm}
\end{table}

\begin{table}[!t]
\centering
\scriptsize
\caption{Experimental results of model utility and $\mathrm{SAPE}$.}
\label{tab-2}
\begin{tabular}{|c|c|c|c|c|}
\hline
\multicolumn{2}{|c|}{Dataset} & Baseline (\%) & Ours (\%) & $\mathrm{SAPE}$ \\ \hline
\multirow{4}{*}{\begin{tabular}[c]{@{}c@{}}Delete\\ (2\%)\end{tabular}} & MNIST & 98.39±0.04 & 98.02±0.01 & 1.884 $\times {10^{ - 3}}$ \\ \cline{2-5} 
 & FMNIST & 90.96±0.02 & 90.02±0.04 & 5.194 $\times {10^{ - 3}}$ \\ \cline{2-5} 
 & CIFAR-10 & 67.33±0.07 & 67.75±0.05 & 3.111 $\times {10^{ - 3}}$ \\ \cline{2-5} 
 & CelebA & 96.41±0.04 & 95.98±0.04 & 2.235 $\times {10^{ - 3}}$ \\ \hline
\multirow{4}{*}{\begin{tabular}[c]{@{}c@{}}Delete\\ (1.5\%)\end{tabular}} & MNIST & 98.42±0.02 & 98.12±0.01 & 1.526 $\times {10^{ - 3}}$ \\ \cline{2-5} 
 & FMNIST & 91.45±0.02 & 91.32±0.05 & 7.113 $\times {10^{ - 4}}$ \\ \cline{2-5} 
 & CIFAR-10 & 68.88±0.04 & 68.74±0.03 & 1.017 $\times {10^{ - 3}}$ \\ \cline{2-5} 
 & CelebA & 97.01±0.04 & 96.50±0.04 & 2.636 $\times {10^{ - 3}}$ \\ \hline
\multirow{4}{*}{\begin{tabular}[c]{@{}c@{}}Delete\\ (1\%)\end{tabular}} & MNIST & 98.98±0.08 & 99.02±0.06 & 2.020 $\times {10^{ - 4}}$ \\ \cline{2-5} 
 & FMNIST & 91.65±0.01 & 91.04±0.08 & 3.339 $\times {10^{ - 3}}$ \\ \cline{2-5} 
 & CIFAR-10 & 69.27±0.09 & 68.75±0.02 & 3.768 $\times {10^{ - 3}}$ \\ \cline{2-5} 
 & CelebA & 97.04±0.02 & 97.11±0.04 & 3.605 $\times {10^{ - 4}}$ \\ \hline
\multirow{4}{*}{\begin{tabular}[c]{@{}c@{}}Delete\\ (0.5\%)\end{tabular}} & MNIST & 98.98±0.08 & 99.02±0.06 & 2.020 $\times {10^{ - 4}}$ \\ \cline{2-5} 
 & FMNIST & 91.67±0.03 & 91.02±0.03 & 3.558 $\times {10^{ - 3}}$ \\ \cline{2-5} 
 & CIFAR-10 & 69.35±0.06 & 68.75±0.05 & 4.345 $\times {10^{ - 3}}$ \\ \cline{2-5} 
 & CelebA & 97.31±0.02 & 97.12±0.03 & 9.772 $\times {10^{ - 4}}$ \\ \hline
\end{tabular}
		\vspace{-0.4cm}
\end{table}

\begin{figure}[!t]
 \centering
 \includegraphics[width=0.65\linewidth]{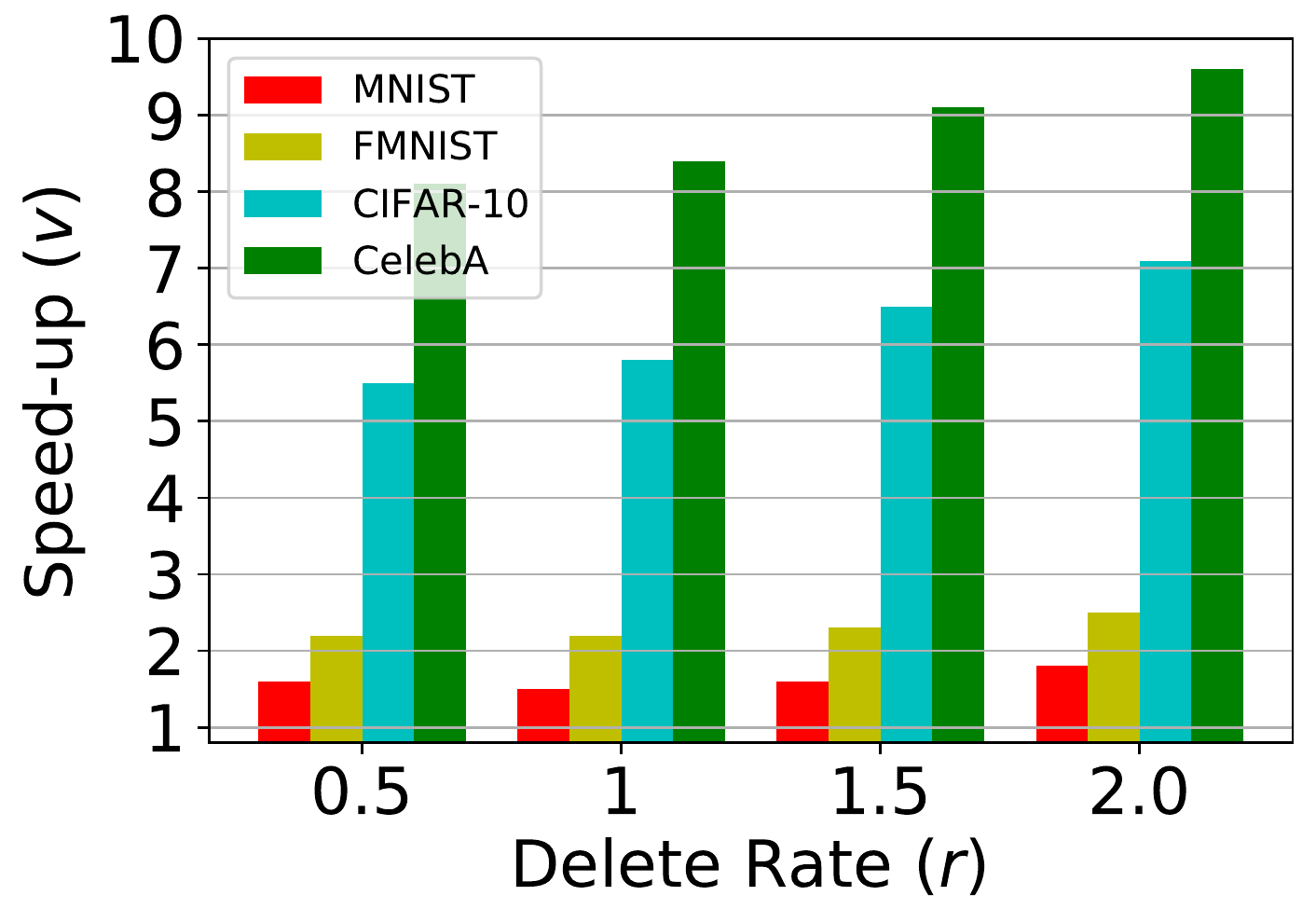}
 \caption{The effect of the data deletion rate $r$ on the unlearning efficiency of the proposed algorithm.}
  \label{fig-2}
  		\vspace{-0.5cm}
\end{figure}

\begin{figure}[!t]  
	\centering
	\subfigure{
	\includegraphics[width=0.41\linewidth]{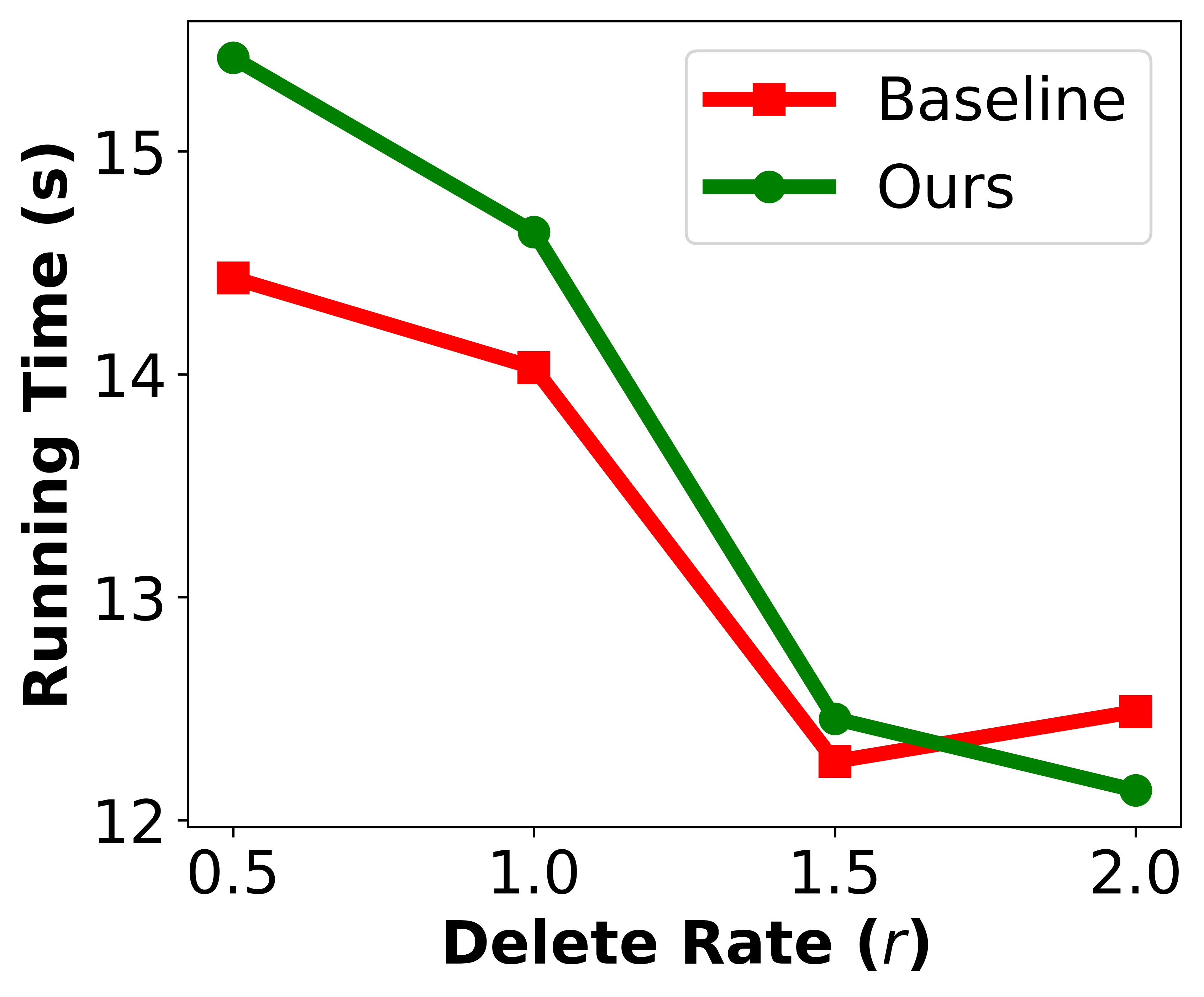}
	}	
	\subfigure{
	\includegraphics[width=0.41\linewidth]{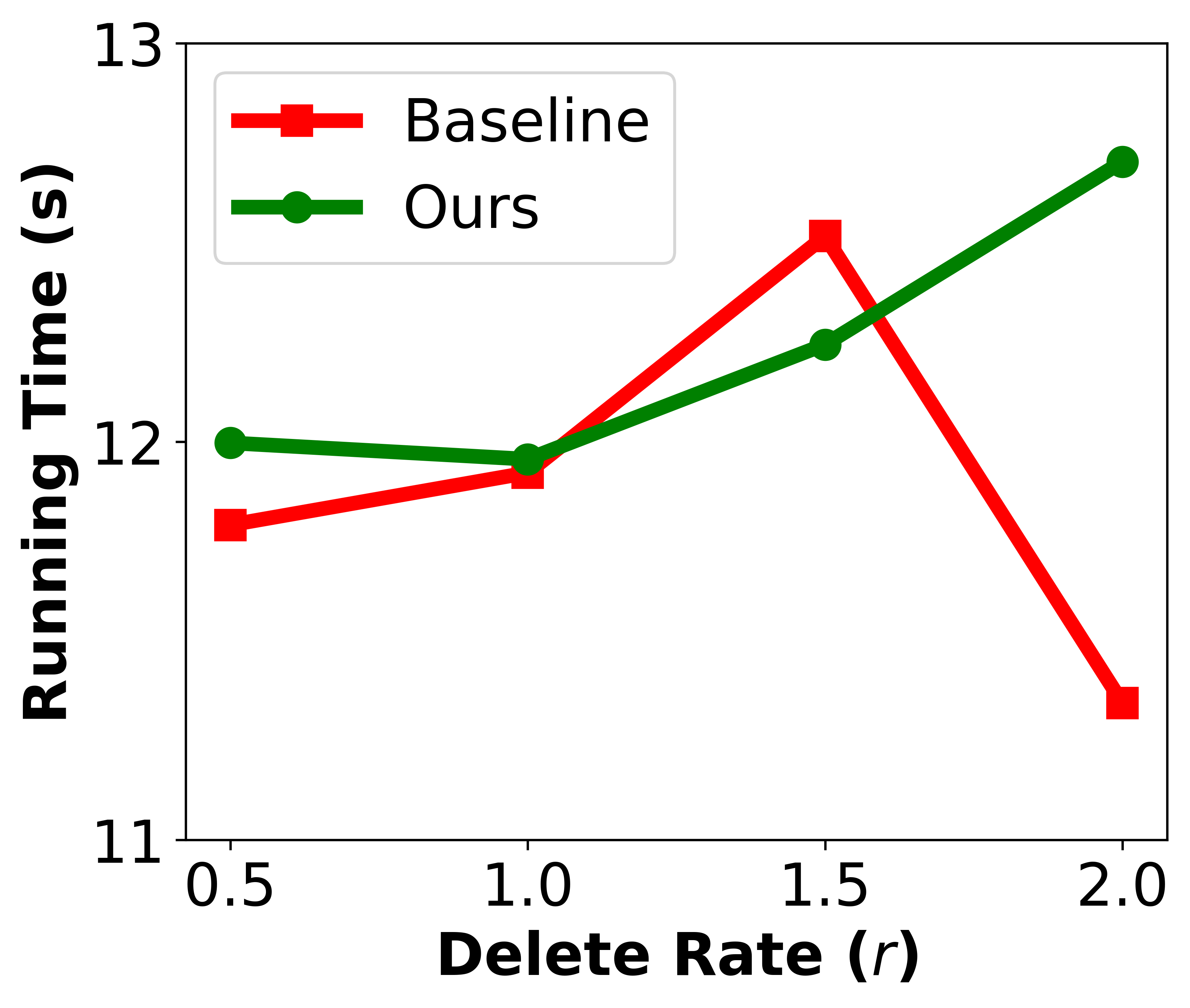}
	}	
	\subfigure{
	\includegraphics[width=0.41\linewidth]{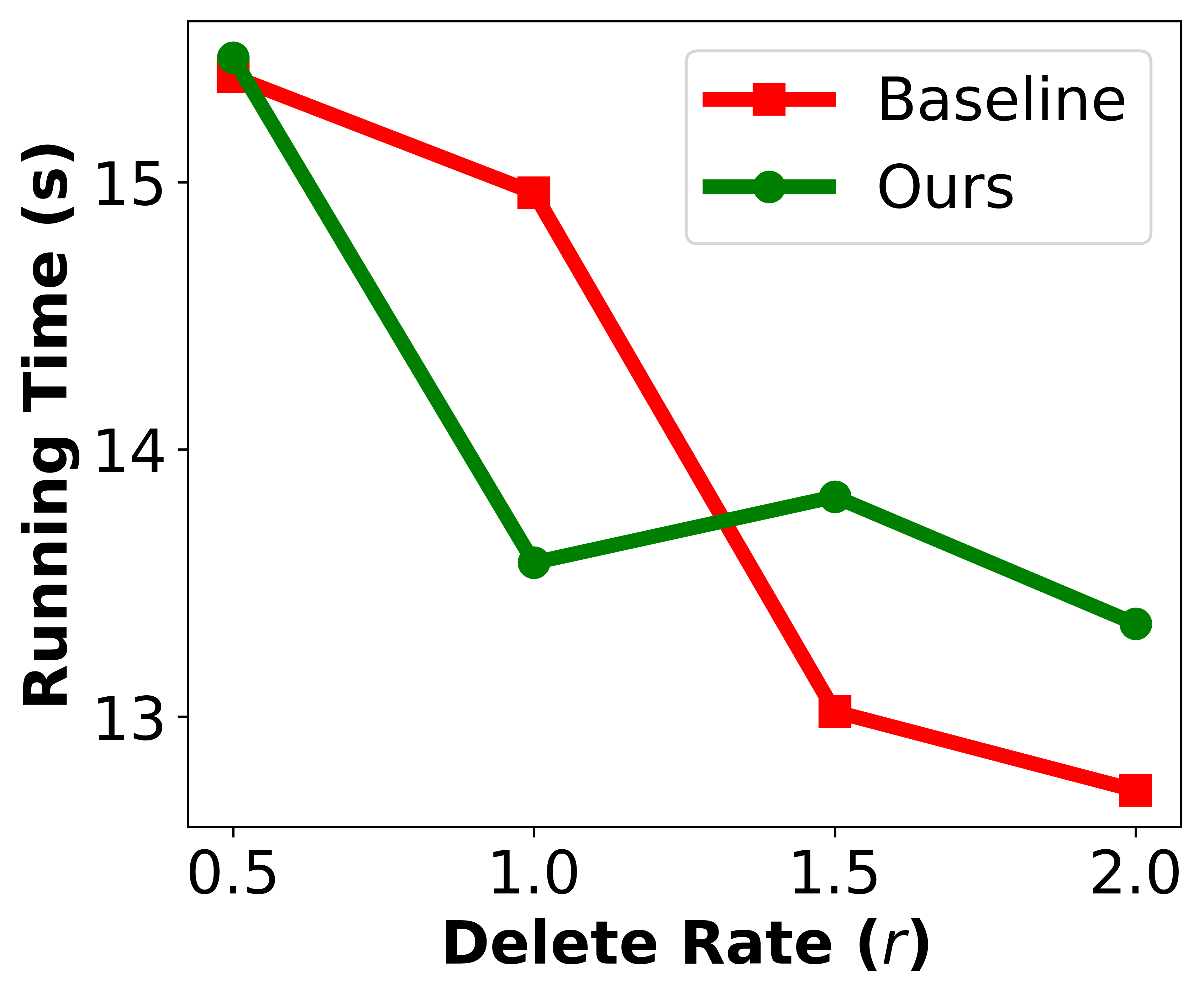}
	}	
	\subfigure{
\includegraphics[width=0.41\linewidth]{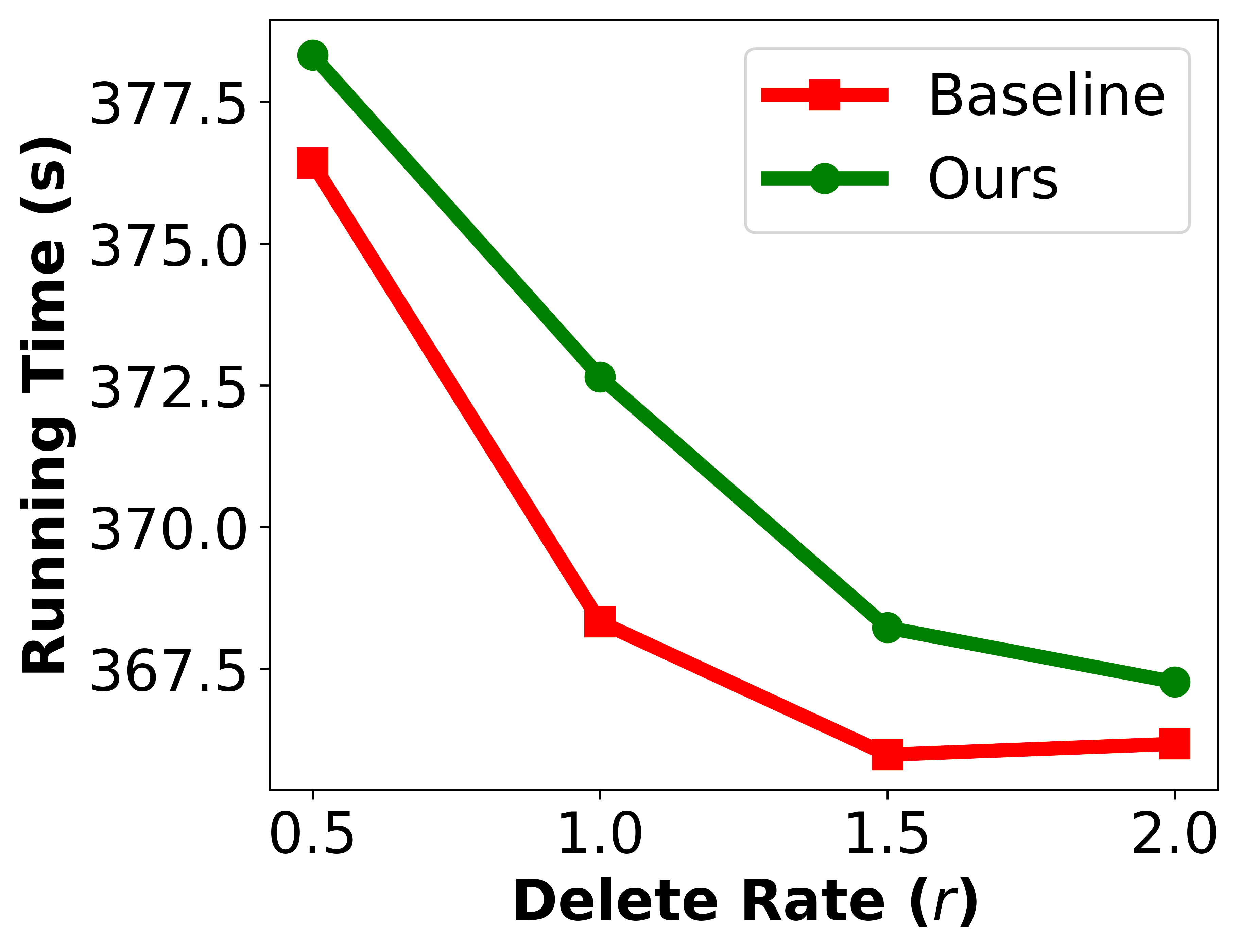}
	}
	\caption{Comparison of the running time of each round of the baseline retraining and our algorithm.}
	\label{fig-3}	
		\vspace{-0.5cm}
\end{figure}

\begin{figure}[!t]  
	\centering
	\subfigure{
	\includegraphics[width=0.41\linewidth]{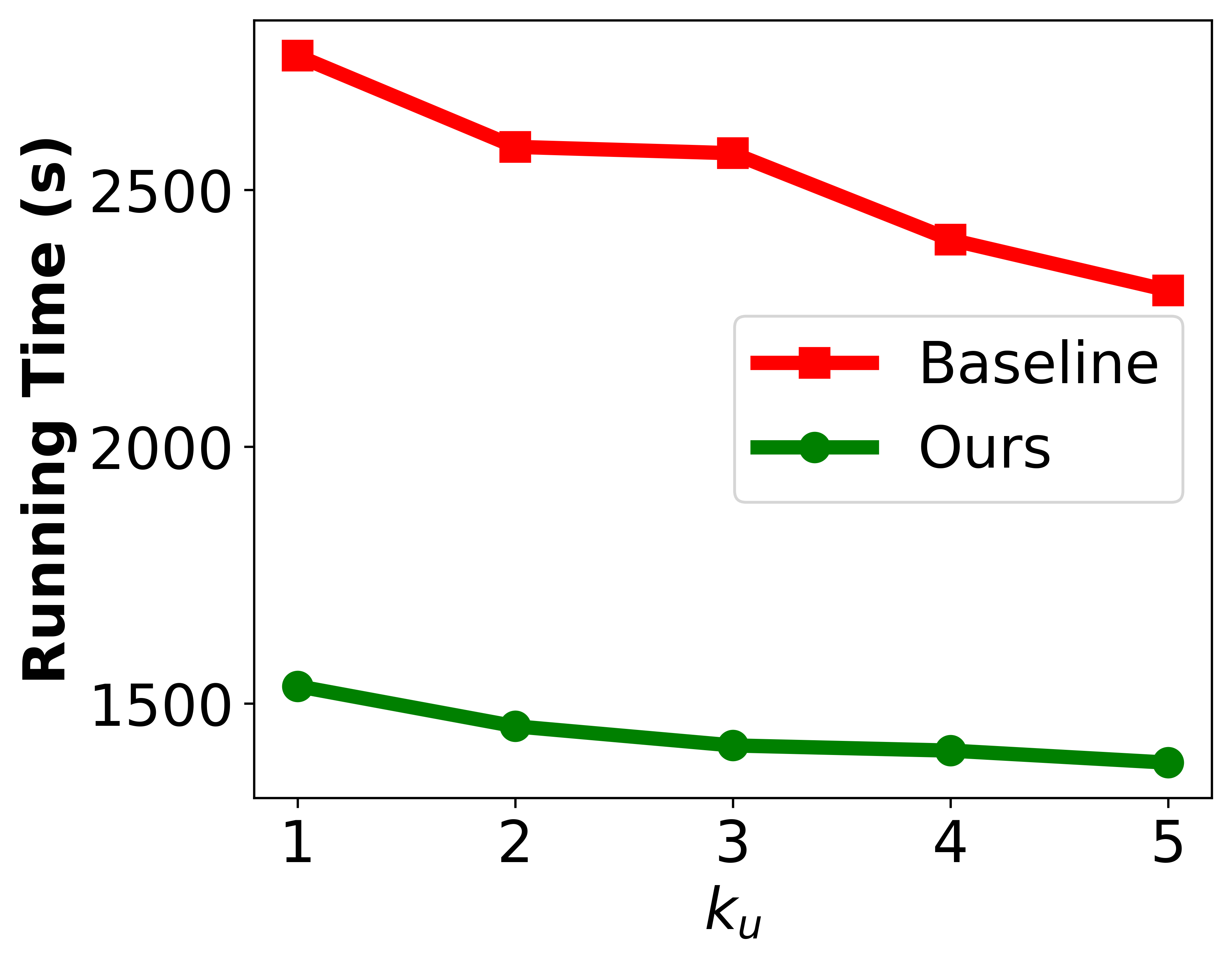}
	}	
	\subfigure{
	\includegraphics[width=0.41\linewidth]{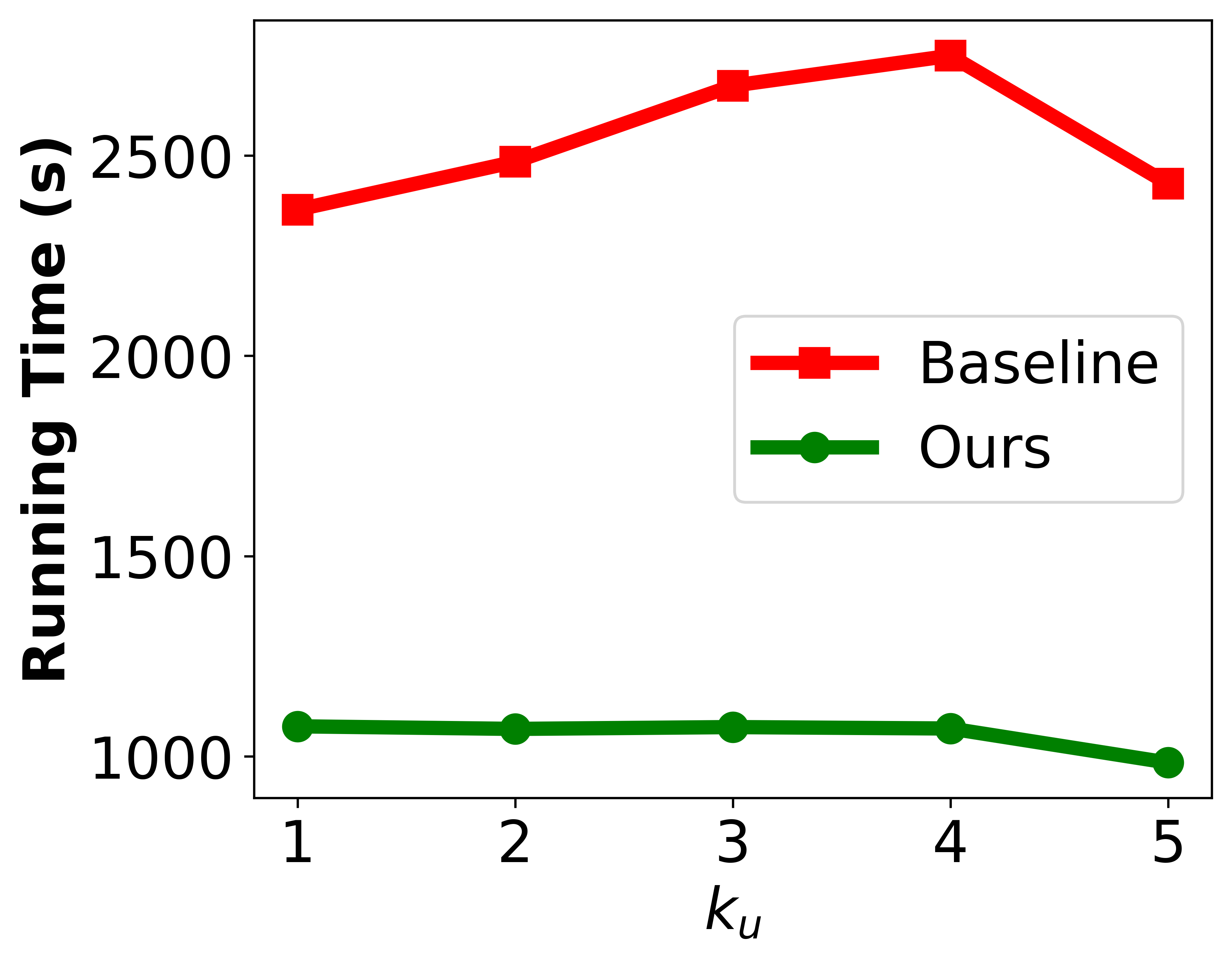}
	}	
	\subfigure{
	\includegraphics[width=0.41\linewidth]{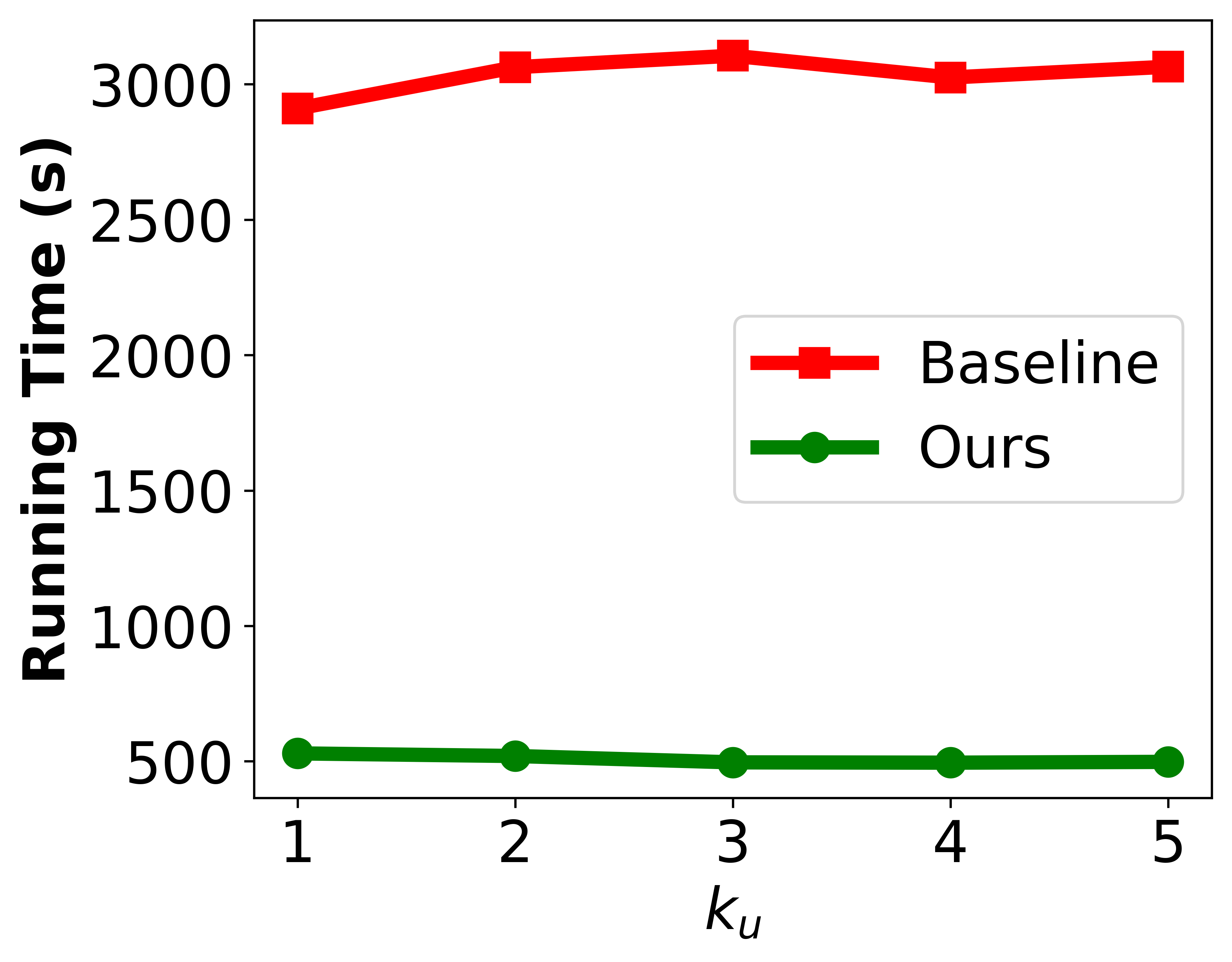}
	}	
	\subfigure{
\includegraphics[width=0.41\linewidth]{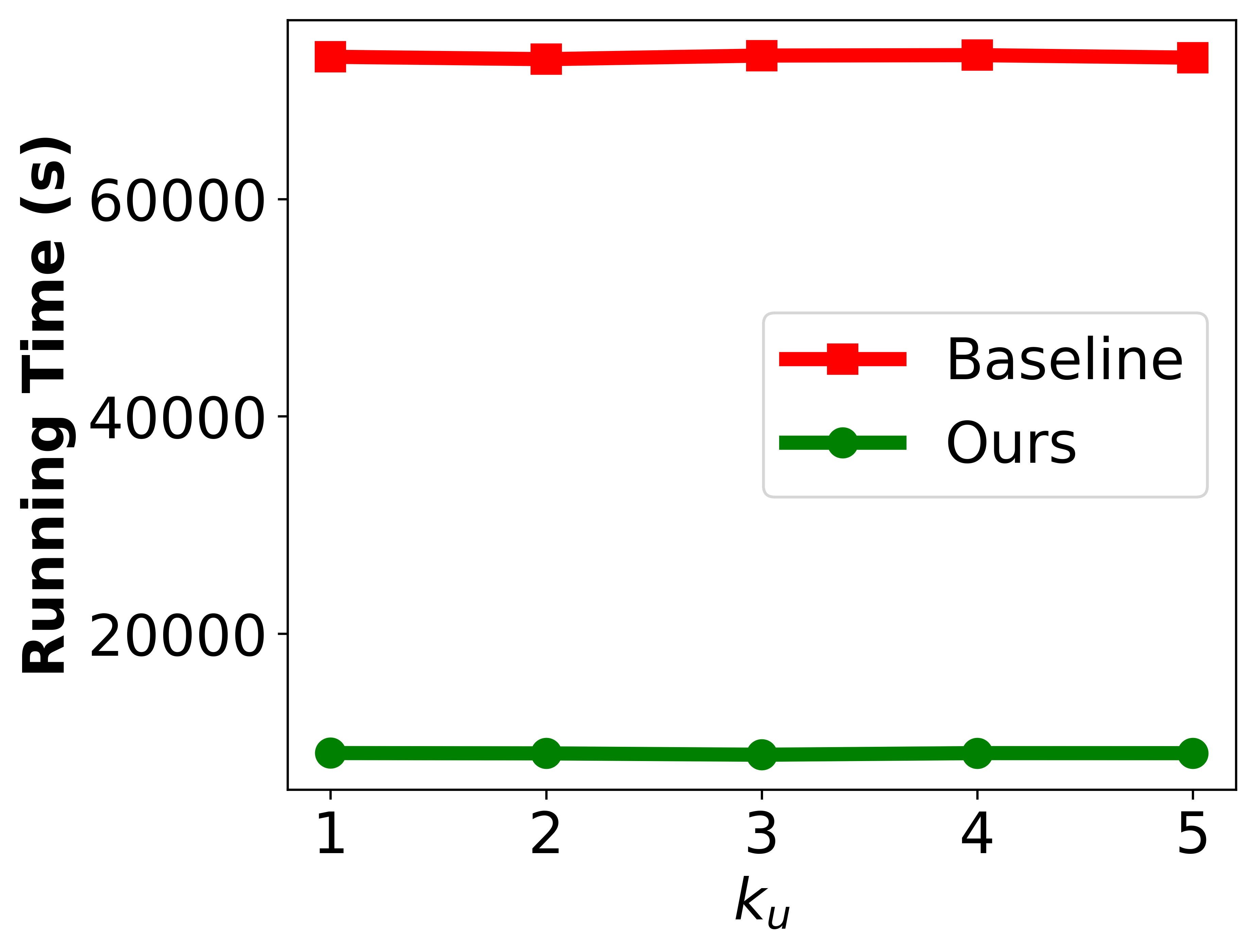}
	}
	\caption{Comparison of the total running time of the baseline retraining and our algorithm under different $k_u$ settings.}
	\label{fig-4}	
		\vspace{-0.5cm}
\end{figure}

\begin{figure}[!t]  
	\centering
	\subfigure{
	\includegraphics[width=0.41\linewidth]{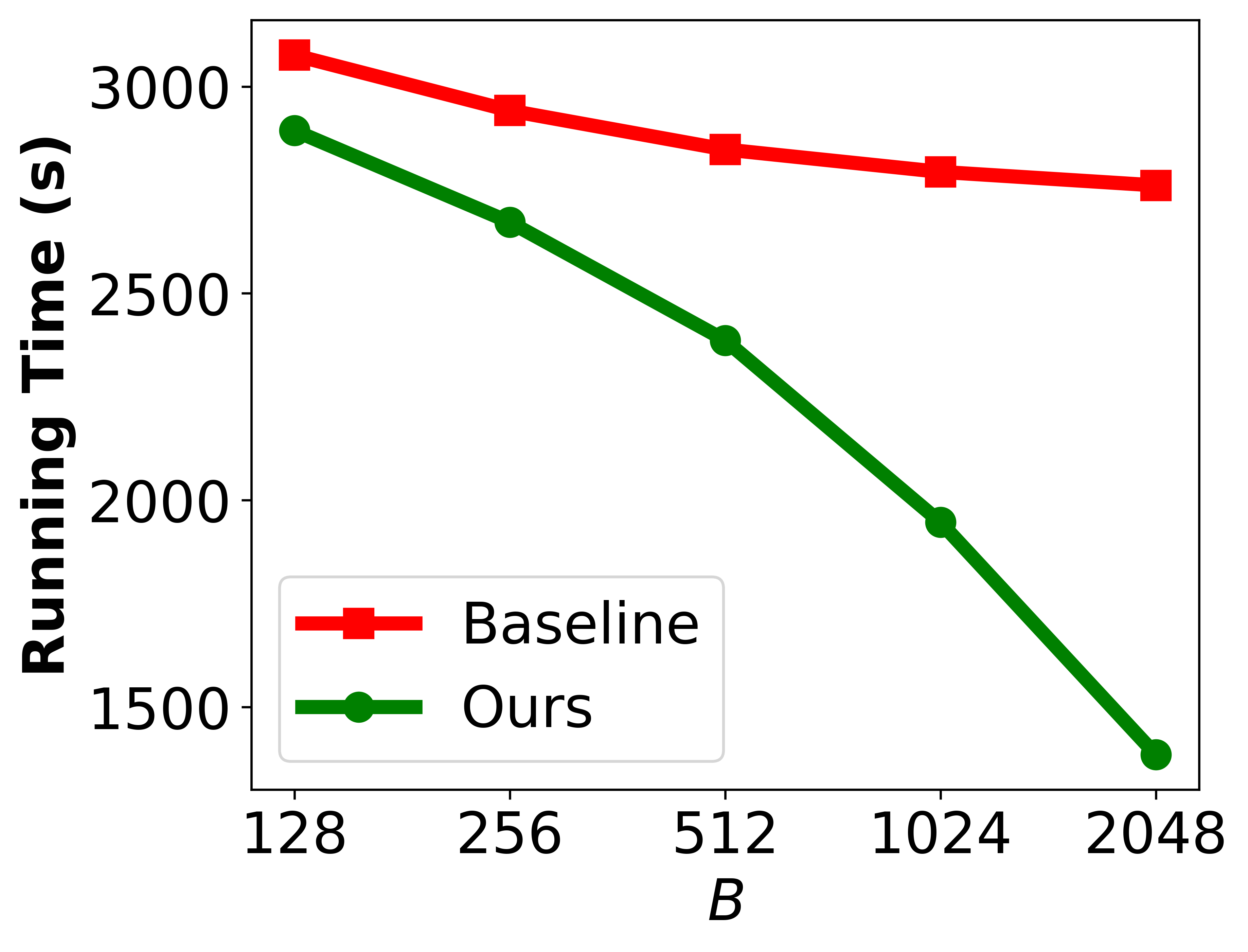}
	}	
	\subfigure{
	\includegraphics[width=0.41\linewidth]{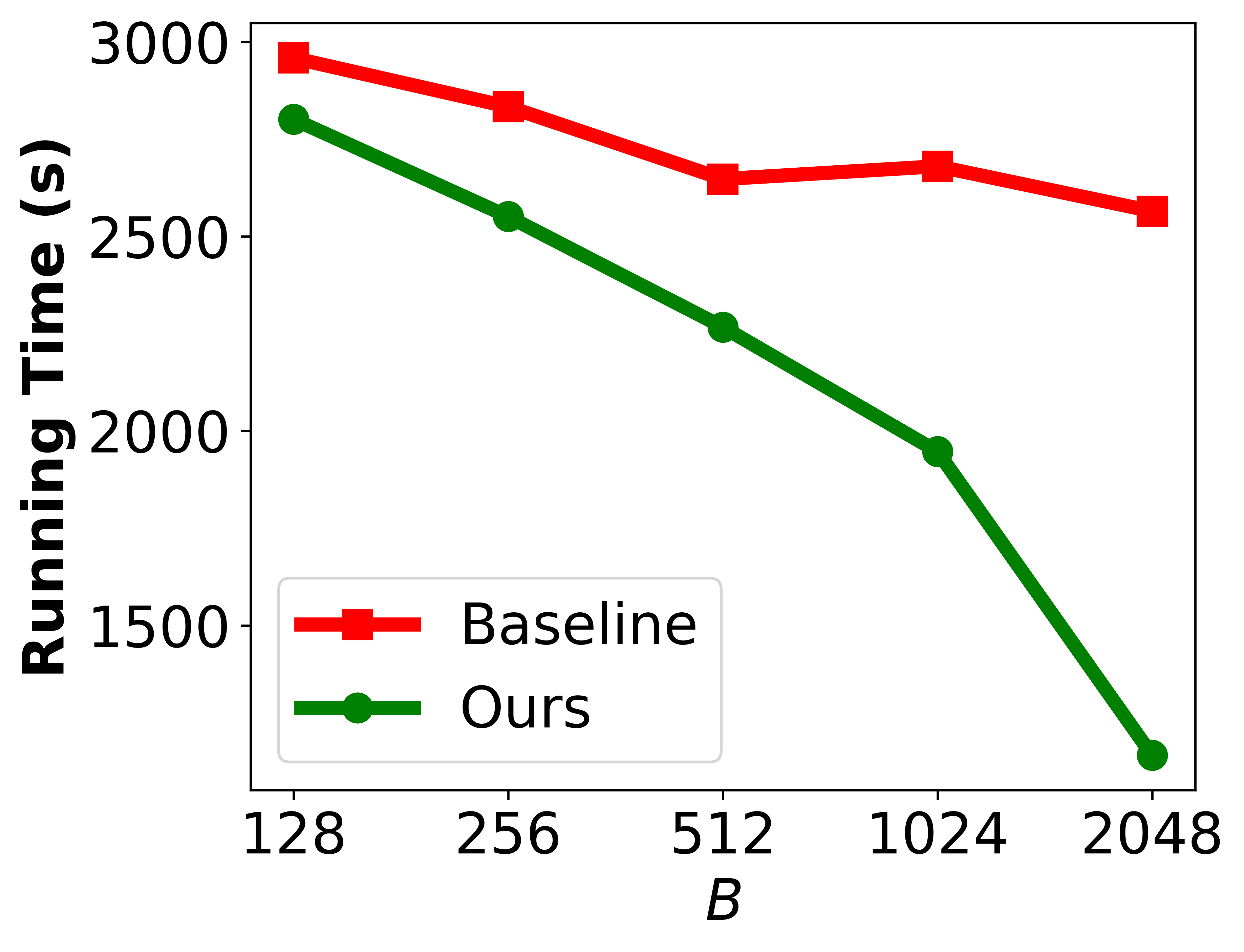}
	}	
	\subfigure{
	\includegraphics[width=0.41\linewidth]{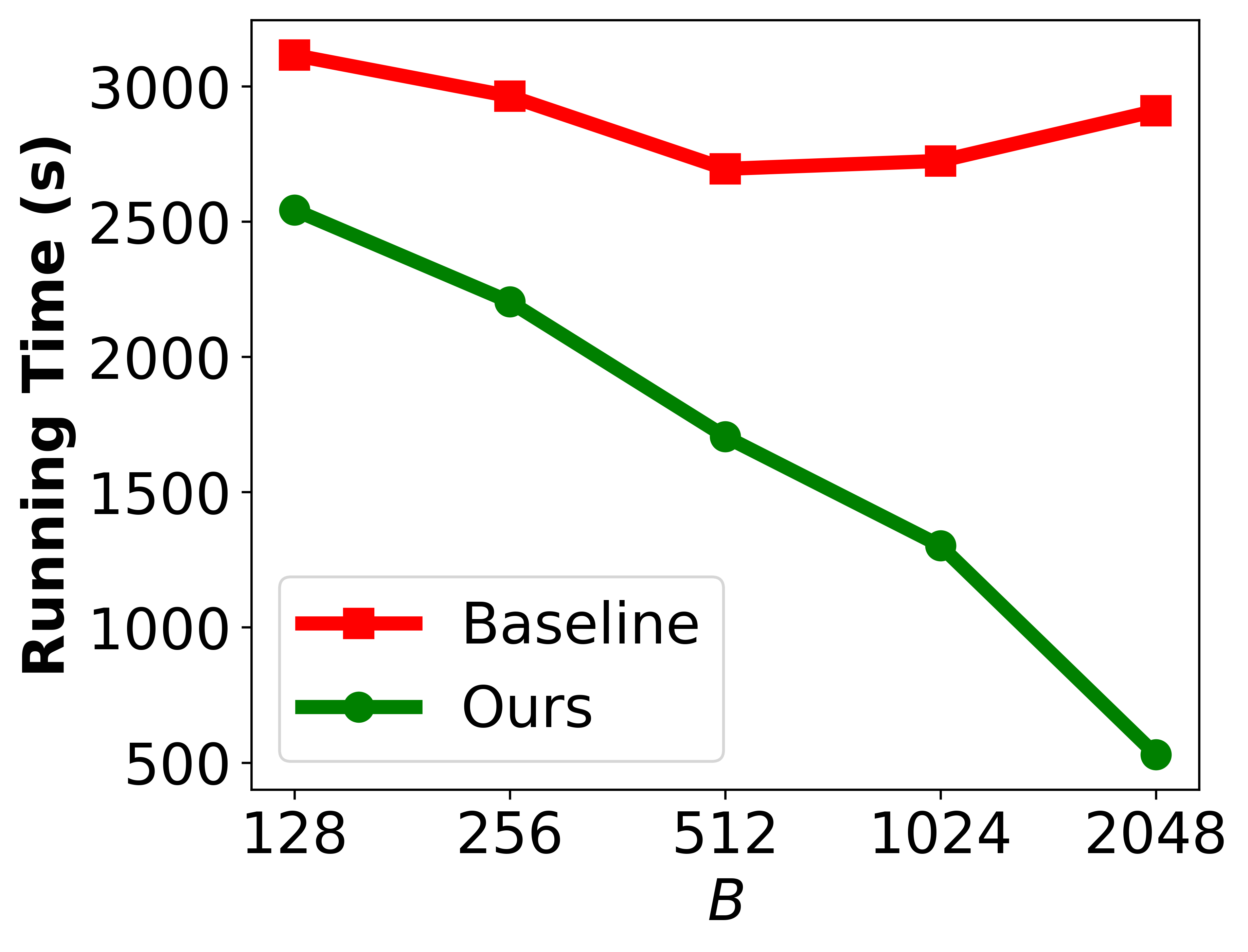}
	}	
	\subfigure{
\includegraphics[width=0.41\linewidth]{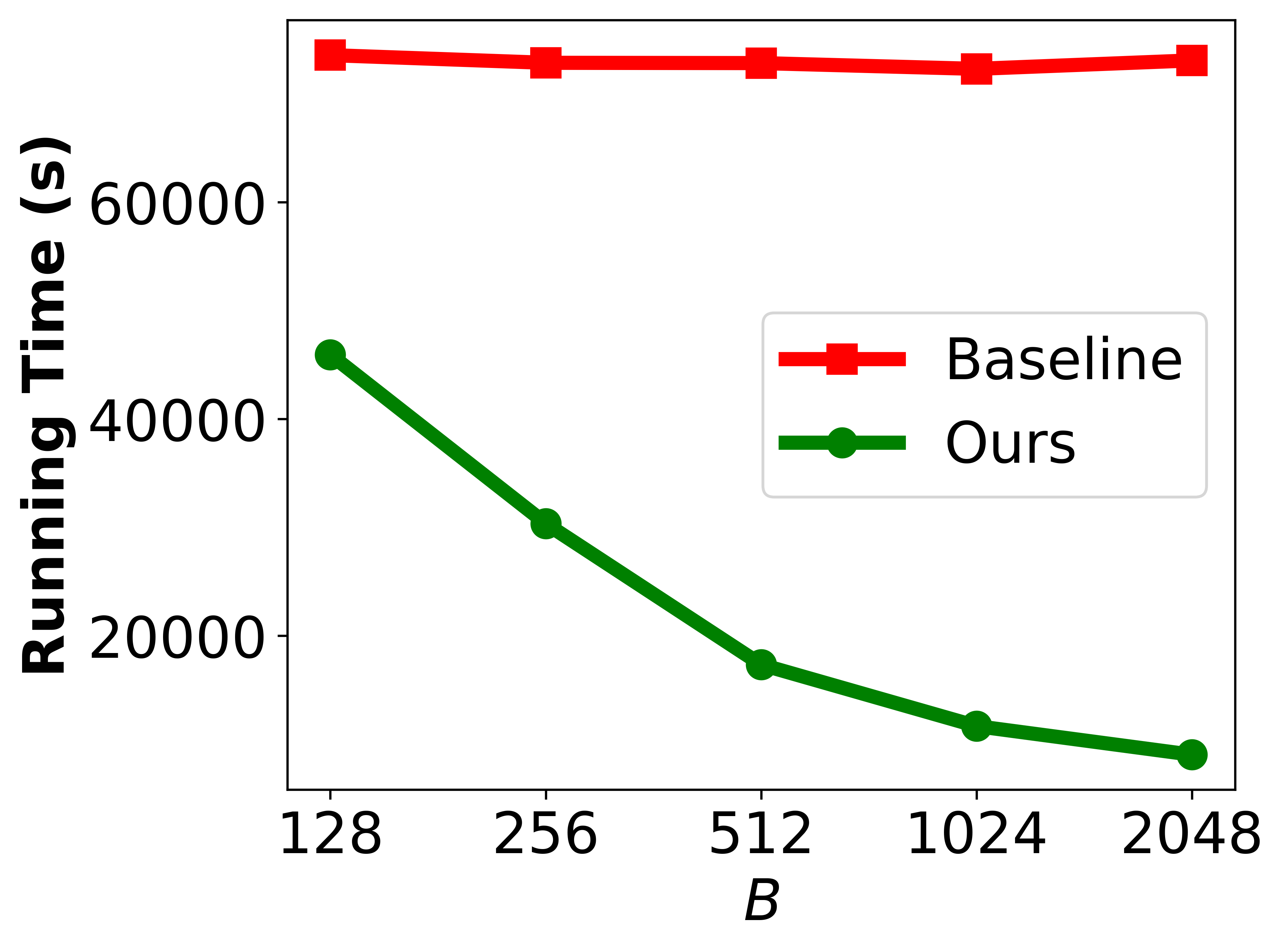}
	}
	\caption{Comparison of the total running time of the baseline retraining and our algorithm under different $B$ settings.}
	\label{fig-5}
	\vspace{-0.5cm}
\end{figure}

\section{Experiments}\label{sec-6}
\subsection{Experiment Setup}\label{sec-6-1}
To evaluate the performance of our proposed design, we conduct extensive experiments on four representative public datasets. All experiments were developed using Python 3.7 and PyTorch 1.7, and executed on a server with an NVIDIA GeForce RTX2080 Ti GPU and an Intel Xeon Silver 4210 CPU. 

\noindent \textbf{Datasets:} In this paper, we adopt four real-world image datasets for evaluations, i.e., MNIST\footnote{\href{http://yann.lecun.com/exdb/mnist/}{http://yann.lecun.com/exdb/mnist/}}, Fashion-MNIST\footnote{\href{https://github.com/zalandoresearch/fashion-mnist}{https://github.com/zalandoresearch/fashion-mnist}}, CIFAR-10\footnote{\href{https://www.cs.toronto.edu/~kriz/cifar.html}{https://www.cs.toronto.edu/~kriz/cifar.html}}, and CelebA\footnote{\href{https://mmlab.ie.cuhk.edu.hk/projects/CelebA.html}{https://mmlab.ie.cuhk.edu.hk/projects/CelebA.html}} (a.k.a. Federated LEAF \cite{caldas2018leaf} dataset). The datasets cover different attributes, dimensions, and number of categories, as shown in Table \ref{tab-1}, allowing us to effectively explore the unlearning utility of the proposed algorithm. To simulate the real environment settings of FL, we evenly distribute the four training datasets to all clients.

\noindent \textbf{Models:} In this experiment, we use a simple CNN model, i.e., CNN with 2 convolutional layers followed by 1 fully connected layer for classification tasks on the MNIST dataset and Fashion-MNIST dataset, the AlexNet model for classification tasks on the CIFAR-10 dataset, and the ResNet-18 model for classification tasks on the CelebA dataset. In particular, we performed a gender classification task on the CelebA dataset.

\noindent \textbf{Hyperparameters:} In our design, we consider the cross-silo FL scenario. We set the number of clients $K=10$, proportion of client participation $q=1$, local epoch ${E_{local}}=1$, mini-batch size $B \in \{ 128,256,512,1024,2048\}$, learning rate $\eta  = 0.001$, delete rate $r = \{ 2\% ,1.5\% ,1\% ,0.5\% \}$, training round $T=200$, the block size $b=3$, and the momentum parameter $\beta_1=0.9$, $\beta_2=0.999$.

\noindent \textbf{Federated Unlearning Pipeline:} First, we follow the above hyperparameters setting to start the training stage in the FL pipeline. Then, we use the obtained model $\omega^*$ for inference or unlearning. Second, if the clients initiate several data erasure requests, the server will start the unlearning stage after receiving all the requests to forget the contribution of the erased data. Notably, we define the percentage of client-side deleted data to all training data as the deletion rate $r$. Finally, the server reinitializes the global model $\omega^*$ and uses Algorithm \ref{al-1} or the baseline algorithm (see below) to perform rapid retraining to obtain the new model $\omega^u$.

\noindent \textbf{Baseline:} \textbf{Retraining from Scratch.} This method is to delete the erased training samples and to retrain the FL model from scratch by using the remaining dataset as the training dataset. 

\noindent \textbf{Evaluation Metrics:} First, we will evaluate the efficiency of the proposed algorithm, and the speed-up of our retraining algorithm running time is defined: $v  = \frac{{{\mathrm{T}_b}}}{{{\mathrm{T}_u}}}\mathrm{x}$. Second, to fairly compare model utility, we report the performance achieved by a given model relative to the performance of the baseline model. To this end, we use the Symmetric Absolute Percentage Error (SAPE) defined as: $\varepsilon_s = \mathrm{SAPE}(Acc_{test}^*,Acc_{test}^u) = \frac{{|Acc_{test}^u - Acc_{test}^*|}}{{|Acc_{test}^*| + |Acc_{test}^u|}}$, where $Acc_{test}^*$ denotes the accuracy of the model $\omega^*$ obtained by the baseline algorithm on the test dataset $\mathcal{D}_{test}$, and $Acc_{test}^u$ denotes the accuracy of the model $\omega^u$ obtained by the proposed algorithm on the same dataset.

\subsection{Evaluation}\label{sec-6-2}
Our evaluation is four-fold: \textit{(i)} To show the data erasure efficiency of Algorithm \ref{al-1}, we report the speed-up of the proposed algorithm; \textit{(ii)} To prove the effectiveness of the proposed algorithm, we compare it with the baseline in terms of accuracy and $\mathrm{SAPE}$ under different delete rates; \textit{(iii)} To explore the parameter sensitivity of the proposed algorithm, we conduct extensive case studies around the parameters $k_u$ and $B$. \textit{(iv)} To answer whether our algorithm can effectively realize unlearning, we conduct an analysis with experimental results. 

First, we evaluate the data erasure efficiency of the proposed algorithm under different deletion rates $r$ settings. Specifically, we fixed the number of unlearned clients $k_u=1$, the learning rate $\eta=0.001$, and batch size $B=2048$. Then, we choose the deletion rate $r = \{ 2\% ,1.5\% ,1\% ,0.5\% \}$ to explore the data erasure efficiency of Algorithm \ref{al-1}. As shown in Fig. \ref{fig-2} and Fig. \ref{fig-3}, we highlight the superiority of the proposed algorithm by computing the speed-up factor $v$ and the running time of each round. For each dataset, it can be clearly seen that the data erasure of the proposed algorithm is always higher than that of the baseline method, which is because that our algorithm utilizes gradient information and curvature information to find a ``better'' descent direction, thereby reducing the time for retraining. We can also see that the data erasure ability of our algorithm is more efficient on large datasets (i.e., CIFAR-10 ($v = 7.1\mathrm{x}$) and CelebA ($v = 9.6\mathrm{x}$). The reason is that if the baseline method uses mini-batch SGD to retrain on a large dataset, this requires multiple recalculations of gradient information for the same sample, which greatly increases the cost of retraining. In Fig. \ref{fig-3}, the running time of each round of the proposed algorithm is very close to that of the baseline, which is in line with the results of the time and space complexity analysis in Section \ref{sec-5-2}. This experiment also shows that Algorithm \ref{al-1} meets the goal \textbf{G4} and \textbf{G5}.

Second, Table \ref{tab-2} shows the model utility of Algorithm \ref{al-1} and the baseline algorithm. In addition, we also report the $\mathrm{SAPE}$ between the accuracy of the unlearned models produced by the two algorithms. Experimental results show that the performance of the proposed algorithm is comparable to that of the baseline, i.e., meet the goal \textbf{G3}. For example, when $r=0.5\%$, the performance of Algorithm \ref{al-1} on the MNIST dataset is slightly higher than the baseline. Furthermore, when we increase $r$, the performance of the proposed algorithm can still be comparable to the baseline. Obviously, the baseline method achieves stable high performance at the expense of retraining costs. Although our algorithm uses diagonal empirical FIM to approximate the Hessian matrix and damages the model utility, the use of the momentum technique makes up for this shortcoming. Similarly, this shows that our method satisfies the goal \textbf{G3}.

Third, we report the sensitivity of Algorithm \ref{al-1} to the parameters $k_u$ and $B$. We use the controlled variable method to analyze the influence of the parameters $k_u$ and $B$ on the data erasure efficiency of the algorithm. Specifically, when we explore the influence of parameter $k_u$, we fix the deletion rate $r=1\%$ and $B=2048$; on the contrary, when exploring the parameter $B$, we fix the deletion rate $r=1\%$ and $k_u=1$. Fig. \ref{fig-4} and Fig. \ref{fig-5} present the total running time of the two algorithms, where Fig. \ref{fig-3} shows that $k_u$ has no obvious effect on speed-up; Fig. \ref{fig-4} shows that $B$ has a significant impact on speed-up. Recall that these experimental results are consistent with the results of the analysis of time and space complexity in Section \ref{sec-5-2}. Notably, the unlearning efficiency of Algorithm \ref{al-1} is significantly higher than the baseline method in the case of large batch size. The reason is that in this case, this algorithm can utilize more samples to compute a more accurate diagonal empirical FIM, thereby quickly realizing unlearning.

Last, we conduct an analysis with experimental results to show the data erasure effectiveness of the proposed algorithm. Concretely, we assume that $\{ ({x_i},{y_i})\} _{i = 1}^n$ is the input sample randomly sampled from the original dataset, where ${x_i} \sim \mathcal{N}(0,{\sigma ^2})$, and $\{ (x_i^u,y_i^u)\} _{i = 1}^n$ is the input sample randomly sampled from the deleted dataset, where ${x_i^u} \sim \mathcal{N}(0,{\sigma ^2})$. Suppose $f(x_i)$ is the output of the global model $\omega^*$ on the original dataset, and $f_u(x_i^u)$ is the output of the global model $\omega^u$ after unlearning on the deleted dataset. If the distribution of model $\omega^u$ and the distribution of model $\omega^*$ are identical, then the following Equation holds: ${d _u} = \sum\limits_{i = 1}^n {||f({x_i}) - {f_u}(x_i^u)||}  \approx 0,n \to \infty .$ In fact, the distance $d_u$ in the above equation is equivalent to the evaluation index $\mathrm{SAPE}$ designed by us. For example, for the CelebA dataset, the $\mathrm{SAPE}$ value of our Algorithm \ref{al-1} is 9.772 $\times 10^{-4}$, which shows that the model generated by this algorithm and the model generated by the baseline are approximately identical. Therefore, these results have shown that the proposed algorithm meets the requirements of Definition \ref{defi-5} and goals \textbf{G1}, and \textbf{G2}.

\section{Conclusion}
In this paper, we present a federated unlearning algorithm, providing efficient and effective unlearning services for data holders in FL. Specifically, we leverage the first-order Taylor expansion approximation technique to customize a rapid retraining algorithm based on diagonal experience FIM. For boosting model utility, we introduce the momentum technique into the unlearning update strategy to further alleviate the negative impact caused by approximation errors. Furthermore, we provide a comprehensive theoretical analysis and experimental support for the proposed algorithm. 
%
%
The performance advantage of our design has been empirically proven on some medium-scale public datasets, revealing its great potential in building data deletion FL services for various applications. 

\section*{ACKNOWLEDGMENT}
This work was supported in part by the Research Grants Council (RGC) of Hong Kong under Grants 11217819, 11217620, 16207818, 16209120, N\_CityU139/21, and R6021-20F; in part by the National Natural Science Foundation of China under Grant 62072240; in part by the Natural Science Foundation of Jiangsu Province under Grant BK20210330; in part by Shenzhen Municipality Science and Technology Innovation Commission under Grant No. SGDX20201103093004019, CityU; and in part by the Australian Research Council (ARC) Discovery Projects under Grant DP200103308.

\bibliographystyle{IEEEtran}

\bibliography{ref}

\clearpage

\end{document}